\documentclass[a4paper,10pt]{article}
\usepackage{amsmath}
\usepackage{amssymb}
\usepackage{amsthm}
\usepackage{amscd}
\usepackage{enumerate}
\usepackage{bm}
\usepackage{geometry}
\usepackage[all]{xy}
\usepackage{tikz-cd}
\usepackage{cite}
\usepackage{hyperref}
\usepackage{lineno}
\usepackage{color}

\usepackage{color}

\newcommand{\p}{\partial}
\newcommand{\ve}{\epsilon}
\newcommand{\al}{\alpha}
\newcommand{\dd}{\mathrm{d}}
\DeclareMathOperator{\Res}{\mathrm{Res}}
\DeclareMathOperator{\nd}{d\!}

\theoremstyle{definition}
\newtheorem{theorem}{Theorem}[section]
\newtheorem*{main}{Main Theorem}
\newtheorem{lem}{Lemma}[section]

\newtheorem{de}{Definition}[section]
\newtheorem{pro}{Proposition}[section]
\newtheorem{cor}{Corollary}[section]
\newtheorem{rem}{Remark}[section]
\numberwithin{equation}{section}
\title{The extended Ablowitz-Ladik hierarchy and a generalized Frobenius manifold}
\author{SI-QI LIU, YUEWEI WANG, YOUJIN ZHANG}
\begin{document}
\maketitle
\begin{abstract}
 We define a certain extension of  the Ablowitz-Ladik hierarchy,  and   prove that  this extended integrable hierarchy coincides with the topological deformation of the Principal Hierarchy of a generalized Frobenius manifold with non-flat unity.
\end{abstract}
\section{Introduction}\label{zh-22}
The Ablowitz-Ladik (AL) equation
\begin{equation}
    ir_{n,t}=r_{n+1}+r_{n-1}-2r_n\pm r_n^*r_n(r_{n+1}+r_{n-1}),
\end{equation}
is a well-known discrete integrable system which was proposed  by Ablowitz and  Ladik in \cite{ablowitz1975nonlinear,M1976Nonlinear}.
It can be represented by the compatibility condition of the following linear system\cite{suris2012problem}:
\begin{equation}
    \Psi_{n+1}=L_n\Psi_n,\quad 
    \Psi_{n,t}=M_n\Psi_n,
\end{equation}
where 
\[L_n=\begin{pmatrix}\lambda&q_n\\r_n&\lambda^{-1}\end{pmatrix},\quad M_n=i\begin{pmatrix}
    \lambda^2-1-q_nr_{n-1}& \lambda q_n-\lambda^{-1}q_{n-1}\\ \lambda r_{n-1}-\lambda^{-1}r_{n}&-\lambda^{-2}+1+q_{n-1}r_n 
\end{pmatrix}
,\quad q_n=\mp r_n^*.\]
One can decompose  $M_n$ as
\begin{align*}
M_n&=M_{n,-}+M_{n,0}+M_{n,+}\\
& =i\begin{pmatrix}
     \lambda^2-q_nr_{n-1}&q_n\lambda\\
     r_{n-1}\lambda &0
 \end{pmatrix}+i\begin{pmatrix}
      -1 &0\\
    0& 1
  \end{pmatrix}+i\begin{pmatrix}
     0&-q_{n-1}\lambda^{-1}\\-r_n\lambda^{-1}&-\lambda^{-2}+q_{n-1}r_n.
 \end{pmatrix}
\end{align*}
and decompose the flow $\frac{\partial}{\partial t}$  as 
\[\frac{\partial}{\partial t}=\frac{\partial}{\partial 
\hat{t}_{0}}+\frac{\partial}{\partial t_{00}} +\frac{\partial}{\partial t_0},\]
where the flows $\frac{\partial}{\partial \hat{t}_0}$, $\frac{\partial}{\partial t_{00}}$, $\frac{\partial}{\partial t_0}$ are given by the following evolution of the wave function  $\Psi_n$: 
 \begin{equation*}
 \Psi_{n,\hat{t}_{0}}=M_{n,-}\Psi_n,\quad\Psi_{n,t_{00}}= M_{n,0}\Psi_{n},\quad \Psi_{n,t_{0}}=M_{n,+}\Psi_n. 
 \end{equation*}
We can define other positive flows $\frac{\partial}{\partial t_k}(k\ge 1)$ and negative flows  $\frac{\partial}{\partial \hat{t}_k}(k\ge1)$  by imposing appropriate evolutionary  conditions on the wave function.  All these flows are mutually commutative and they form the AL hierarchy, as it is introduced in \cite{suris2012problem}.   
We define the variables $q(x)$, $r(x)$ by applying an $\epsilon$-interpolation to the discrete variables of $q_n$, $r_n$ such that 
\[q_n=q(x)|_{x=n\epsilon},\quad r_n(t)=r(x)|_{x=n\epsilon},\]
then the AL hierarchy can be rewritten as evolutionary equations of $q(x)$, $r(x)$.

Now we define a pair of new unknown functions\cite{li2022tri, suris2012problem} 
\begin{equation}
    P=\frac{q}{q^-},\quad Q=\frac{q}{q^{-}}(1-q^{-}r^{-}),
\end{equation}
here and in what follows we denote $f^{+}(x)=\Lambda f(x)$,
$f^{-}(x)=\Lambda^{-1} f(x)$ with $\Lambda=e^{\epsilon\partial_x}$. In terms of the functions $(P,Q)$ the positive and negative flows of the AL hierarchy can be represented by the following Lax equations:
\begin{align}
\frac{\partial L}{\partial t_{k}}&=\frac{1}{\ve(k+1)!}[(L^{k+1})_{+},L],\quad k\ge 0\label{lax1}\\
    \frac{\partial L}{\partial \hat{t}_{k}}&=\frac{(-1)^k k!}{\ve}[(M^{k+1})_{-},L], \quad k\ge 0\label{lax2}
\end{align}
where the Lax operators $L$ and $M$ have the forms 
\begin{align}\label{laxop1}
    L&=(1-Q\Lambda^{-1})^{-1}(\Lambda-P)=\Lambda+(Q-P)+Q(Q^{-}-P^{-})\Lambda^{-1}+\dots\\
    M&=(\Lambda-P)^{-1}(1-Q\Lambda^{-1})=\frac{Q}{P}\Lambda^{-1}+\frac{Q^+}{PP^+}-\frac{1}{P}+(\frac{Q^{++}}{PP^+P^{++}}-\frac{1}{PP^+})\Lambda+\dots\label{laxop2}
\end{align}
This integrable hierarchy is called the relativistic Toda hierarchy in the literature, see for examples \cite{suris2012problem,oevel1989mastersymmetries}, nonetheless we still call it the AL hierarchy.

The purpose of the present paper is to study the relationship of the AL hierarchy with a certain generalized Frobenius manifold. Our main motivation comes from Brini's work on the equivariant Gromov-Witten invariants of the resolved conifold
\[X=\mathcal{O}_{\mathbb{P}^1}(-1)\oplus \mathcal{O}_{\mathbb{P}^1}(-1)\] 
with anti-diagonal $C^*$-action on the fibers. He conjectured in \cite{brini2012local} that the generating function of these Gromov-Witten invariants yields a tau function of the AL hierarchy, and proved his conjecture at the genus one approximation. Brini and his collaborators further showed in \cite{Brini2011Integrable} that the dispersionless limit of the positive flows of the AL hierarchy belong to the Principal Hierarchy of a generalized Frobenius manifold $M_\textrm{AL}$ with non-flat unity, which is almost dual to the Frobenius manifold $M_{X}$ associated with the Gromov-Witten invariants of the resolved conifold with anti-diagonal action. This generalized Frobenius manifold has the potential 
\begin{equation}\label{zh-4}
F=\frac12 (v^1)^2 v^2+v^1 e^{v^2}+\frac12 (v^1)^2 \log v^1
\end{equation}
and flat metric 
\begin{equation}\label{zh-5}
\eta=(\eta_{\al\beta})=\begin{pmatrix} 0 & 1\\ 1 &0\end{pmatrix}.
\end{equation}
We are to construct a certain extension of the AL hierarchy by including to it some additional flows, and to show that the extended AL hierarchy is the topological deformation of the Principal Hierarchy of the generalized Frobenius manifold $M_\textrm{AL}$.

Recall that the Principal Hierarchy of a Frobenius manifold $M^n$ is a bihamiltonian integrable hierarchy 
\begin{equation}\label{zh-1}
\frac{\p v^\al}{\p t^{\beta,q}}=\eta^{\al\gamma}\frac{\p}{\p x}\left(\frac{\p\theta_{\beta,q+1}(v)}{\p v^\gamma}\right),\quad \al,\beta=1,\dots,n,\, q\ge 0
\end{equation}
of hydrodynamic type with spatial variable $x$, where $\eta=(\eta^{\al\beta})^{-1}$ is the flat metric of the Frobenius manifold, the unknown functions $v^1,\dots, v^n$ are given by a system of flat coordinates of $\eta$, and the functions $\theta_{\beta, q}(v)=\theta_{\beta,q}(v^1,\dots,v^n)$ are defined via the deformed flat coordinates of the Frobenius manifold \cite{Du-94, normal}.
Here and in what follows we assume summation from over repeated upper and lower Greek indices which range from 1 to n. In the case when the Frobenius manifold arises from the Gromov-Witten invariants of a smooth projective variety $X$, the logarithm of the tau function $\tau^{[0]}$ of a particular solution 
\[v^\al=\eta^{\al\gamma}\frac{\p^2\log\tau^{[0]}}{\p t^{\gamma,0}\p t^{1,0}},\quad \al=1,\dots,n \]
of the Principal Hierarchy coincides with the generating function of the genus zero Gromov-Witten invariants of $X$. In \cite{normal} Dubrovin and the third author of the present paper constructed a certain dispersionful deformation, called the topological deformation, of the Principal Hierarchy of a semisimple Frobenius manifold. Such a deformation of the Principal Hierarchy is realized via a deformation of the tau function
\begin{equation}\label{zh-2}
\log\tau^{[0]}\to \log\tau=\ve^{-2}\log\tau^{[0]}+\sum_{k\ge 1}\ve^{k-2} \mathcal{F}_k(v,v_x,\dots,v^{(m_k)}),
\end{equation}
where the functions $\mathcal{F}_k$ depend on the unknown functions $v^1,\dots,v^n$ and their $x$-derivatives up to certain orders $m_k$, they are uniquely determined (up to addition of constants) by the condition that the action of the Virasoro symmetries $\frac{\p}{\p s_m},\, m\ge -1$ of the Principal Hierarchy of the Frobenius manifold act linearly on the deformed tau function, i.e. they can be represented as the actions of some linear operators on the tau function as follows:
\[\frac{\p\tau}{\p s_m}=L_m\tau,\quad m\ge -1,\]
where $L_m$ are the Virasoro operators of the Frobenius manifold. This condition of the linearization of the Virasoro symmetries is equivalent to a system of linear equations for the gradients of the function
\begin{equation}\label{4.2.1}
\Delta\mathcal{F}=\sum_{k\ge 1}\ve^{k-2} \mathcal{F}_k(v,v_x,\dots,v^{(m_k)}),
\end{equation}
which is called the loop equation of the Frobenius manifold.
The unique solution of the loop equation yields a quasi-Miura transformation
\begin{equation}\label{zh-3}
v^\al\to w^\al=v^\al+\frac{\p^2}{\p t^{1,0}\p t^{\al,0}}\left(\sum_{k\ge 1}\ve^{k} \mathcal{F}_k(v,v_x,\dots,v^{(m_k)})\right),\quad \al=1,\dots,n
\end{equation}
for the Principal Hierarchy of the semisimple Frobenius manifold, and the resulting equations for the unknown functions $w^1,\dots, w^n$ form the topological deformation of the Principal Hierarchy.

An analogue of the above-mentioned construction of integrable hierarchy is proposed in \cite{liu2022generalized} for any semisimple generalized Frobenius manifold with non-flat unity. We know that for a usual Frobenius manifold the unit vector field $e$ for the multiplication is flat with respect to the flat metric, and one usually chooses the flat coordinates $v^1,\dots, v^n$ such that $e=\frac{\p}{\p v^1}$. The generalized Frobenius manifolds considered in \cite{liu2022generalized} satisfy all the axioms given in Dubrovin's definition of Frobenius manifolds but the flatness of the unit vector field with respect to the flat metric. One difference of the construction of integrable hierarchy for such a generalized 
Frobenius manifold from the above-mentioned one for a usual Frobenius manifold is that one needs to introduce a set of additional flows 
\[\frac{\p v^\al}{\p t^{0,q}}=A^{\al}_{\gamma,q}(v)v^\gamma_x,\quad \al=1,\dots,n,\, q\in\mathbb{Z}\]
to the integrable hierarchy \eqref{zh-1}. These flows together with the ones given in \eqref{zh-1} form the Principal Hierarchy of the generalized Frobenius manifold.
It also possesses a tau structure and Virasoro symmetries. 
By using the condition of linearization of the actions of the Virasoro symmetries on the tau function, one obtains the topological deformation of the Principal Hierarchy via a deformation of the tau function of the form \eqref{zh-2} and a quasi-Miura transformation of the form \eqref{zh-3}.  

It is conjectured in \cite{liu2022generalized}, under a certain Miura-type transformation, that the positive flows $\frac{\p}{\p t_k},\, k\ge 0$ of the AL hierarchy coincide with the flows $\frac{\p}{\p t^{2,k}},\, k\ge 0$ of the topological deformation of the Principal Hierarchy of the generalized Frobenius manifold $M_\textrm{AL}$ with potential \eqref{zh-4}.

We have the following theorem.
\begin{main}
Let $\tau$ be a tau-function of the topological deformation of the Principal hierarchy of the generalized Frobenius manifold $M_\textrm{AL}$, then the functions $P, Q$ defined by 
\begin{equation}
    Q-P=\epsilon(\Lambda-1)\frac{\partial \log \tau}{\partial t^{2,0}},\quad \log Q-\log P=(1-\Lambda^{-1})(\Lambda-1)\log \tau 
\end{equation}
satisfy the Lax equations \eqref{lax1}, \eqref{lax2} if we identify $t_{k}$ and $\hat{t}_k$ with $t^{2,k}$ and $t^{0,-k-1}$ respectively.
\end{main}

The paper is organized as follows.
In Sect.\,2, we construct the Principal Hierarchy of the generalized Frobenius manifold $M_\textrm{AL}$ and its tau structure. In Sect.\,3, we give a certain extension of the Ablowitz-Ladik hierarchy and its tau structure. In Sect.\,4, we introduce a super tau cover of the extended AL hierarchy. In Sect.\,5, we prove that the extended Ablowitz-Ladik hierarchy possesses an infinite number of Virasoro symmetries which act linearly on its tau function, and we prove the Main Theorem. Sect.\,6 is for a conclusion.

\section{The Principal Hierarchy of $M_\textrm{AL}$}
In this section, we construct the Principal Hierarchy of $M_\textrm{AL}$ and its tau structure following the definitions of  \cite{liu2022generalized}, and show that the dispersionless limit of the positive and negative flows \eqref{lax1}, \eqref{lax2} of the AL hierarchy belong to this Principal Hierarchy.

The Frobenius manifold structure of $M_\textrm{AL}$ with potential \eqref{zh-4} and flat metric \eqref{zh-5} is obtained in \cite{Brini2011Integrable} by using the superpotential  
\begin{equation}\label{zh-7}
    \lambda(p)=p+v^1+v^1 e^{v^2}(p-e^{v^2})^{-1},
\end{equation}
here $v^1, v^2$ are the flat coordinates of the flat metric. The non-flat unity $e$ and the Euler vector field $E$ of $M_\textrm{AL}$ are given by 
\begin{equation}
    e=\frac{v^1\partial_{v^1}-\partial_{v^2}}{v^1-e^{v^2}}, \quad E=v^1\partial_{v^1}+\partial_{v^2},
\end{equation}
they are gradients of the function 
\begin{equation}\label{zh-6} 
\varphi=v^2-\log(e^{v^2}-v^1)
\end{equation}
with respect to the flat metric $\eta$ and the intersection form $g$ of $M_\textrm{AL}$ respectively, here 
\[g=(g^{\al\beta})=\begin{pmatrix}
    2v^1e^{v^2} &v^1+e^{v^2}\\
    v^1+e^{v^2} &2
    \end{pmatrix}.\]
    
On $M_\textrm{AL}$ we have the deformed flat connection $\widetilde{\nabla}$ defined by 
\[\widetilde{\nabla}_{a}b=\nabla_{a}b+za\cdot b,\quad \forall a,b\in \textrm{Vect}(M),\]
where $\nabla$ is the Levi-Civita connection of the flat metric $\eta$, and the multiplication $a\cdot b$ of the vector fields $a, b$ is defined by
$\p_\al\cdot\p_\beta=c^\gamma_{\al\beta}\p_\gamma$ with $\p_\al=\frac{\p}{\p v^\al}$ and 
\[c_{\al\beta}^\gamma=\eta^{\gamma\xi}\frac{\p^3 F(v)}{\p v^\al\p v^\beta\p v^\xi},\quad (\eta^{\al\beta})=(\eta_{\al\beta})^{-1}.\]
This deformed connection can be extended to a flat connection on $M_\textrm{AL}\times \mathbb{C}^*$ by the formulae
\begin{align*}
\widetilde{\nabla}_a\frac{\mathrm{d}}{\mathrm{dz}}=0,\quad \widetilde{\nabla}_{\frac{\mathrm{d}}{\mathrm{dz}}}\frac{\mathrm{d}}{\mathrm{dz}}=0,\quad \widetilde{\nabla}_{\frac{\mathrm{d}}{\mathrm{dz}}}b=\partial_z b+E\cdot b-\frac{1}{z}\mu v,
\end{align*}
here $a, b$ are vector fields with zero  components along $\mathbb{C}^*$, $z$ is the coordinate on $\mathbb{C}^*$, and $\mu$ is the grading operator defined by
\begin{equation}\label{mu}
    \mu=\frac{2-d}{2}-\nabla E=\textrm{diag}\left(-\frac{1}{2},\frac{1}{2}\right)
\end{equation}
 with $d=1$ being the charge of $M_\textrm{AL}$. A basis of horizontal section of the connection $\widetilde{\nabla}$ yields a system of deformed flat coordinates $(\tilde{v}^1,\tilde{v}^2)$ of $M_\textrm{AL}$, which can be chosen to take the form
 \cite{Du-94, dubrovin1999painleve, normal, liu2022generalized}
\begin{equation}
    (\tilde{v}_1,\tilde{v}_2)=(\theta_{1}(v,z),\theta_2(v,z))z^{\mu}z^{R},
\end{equation}
where the constant matrix
\[R=R_1=\begin{pmatrix}
    0&0\\2&0
\end{pmatrix}\]
is part of the monodromy data of $M_\textrm{AL}$ at $z=0$, and the functions $\theta_{\al}(v,z)$ are analytic at $z=0$ and have the series expansions
\[\theta_{\al}(v,z)=\sum_{k\ge 0}\theta_{\al,k}(v)z^{k},\quad \theta_{\al,0}(v)=v_\al:=\eta_{\al\gamma}v^\gamma,\quad \al=1,2.\]
The coefficients $\theta_{\al,k}(v)$ satisfy the equations
\begin{align}
&\p _{\gamma}\p_{\beta}\theta_{\al,k+1}=c_{\gamma\beta}^{\ve}\p_{\ve} \theta_{\al,k},\quad \al,\beta,\gamma=1,2,\,k\ge 0;\label{re2}\\
&\p_E\nabla\theta_{\al,k}=(k+\mu_{\al}-\mu)\nabla\theta_{\al,k}+(R_1)^{\ve}_{\al}\nabla\theta_{\ve,k-1},\quad \al,\gamma=1,2,\,k\ge 1.\label{quasi-homo}
\end{align}
By integrating the relations \eqref{quasi-homo} and taking the integration constants to be zero, we arrive at
\begin{equation}\label{quasi-2}
\p _E \theta_{1,k}=k\theta_{1,k}+2\theta_{2,k-1},\quad \p_E \theta_{2,k}=(k+1)\theta_{2,k},\quad k\ge 1.
\end{equation}
Then the recursion relations \eqref{re2} and the quasi-homogenous condition 
\eqref{quasi-2} uniquely fix the functions $\theta_{\al,k},\,\al=1, 2$.
The first few of them are given by
\begin{align}
   & \theta_{1,1}=(v^2 + \log v^1)v^1 + (e^{v^2}-v^1), \quad\theta_{2,1}=v^1e^{v^2}+\frac{1}{2}(v^1)^2,\label{zh-20}\\
   & \theta_{1,2}=(v^2+\log v^1)\theta_{2,1}+\frac{1}{4}(e^{2v^2}-4v^1e^{v^2}-(v^1)^2),\\
   &\theta_{2,2}=\frac{1}{2}v^{1}e^{2v^2}+(v^1)^2e^{v^2}+\frac{1}{6}(v^1)^3.
\end{align}

To define the Principal Hierarchy of $M_\textrm{AL}$ we need to introduce, following \cite{liu2022generalized}, another set of functions $\theta_{0,\ell}$ for $\ell\in\mathbb{Z}$. 
They are determined by $\theta_{0,0}=\varphi$ and the relations
\begin{align}
&\p_{\alpha}\p _{\beta}\theta_{0,\ell}=c_{\al\beta}^{\gamma}\p_{\gamma}\theta_{0,\ell-1},\quad \al,\beta=1,2,\,\ell\in\mathbb{Z};\label{re0-2}\\
&\p_E\theta_{0,k}=k\theta_{0,k}+\theta_{2,k-1},\quad 
\p_E \theta_{0,-k}=-k\theta_{0,-k},\quad k\ge1.\label{quasi0}
\end{align}
Here the function $\varphi$ is defined by \eqref{zh-6}. The first four functions $\theta_{0,\ell}$ have the expressions
    \begin{align*}
 \theta_{0,1}&=v^1v^2,\quad \theta_{0,2}=\frac{1}{2}(v^2 +1)(v^1)^2 + (v^2-1)e^{v^2}v^1,\\ \theta_{0,-1}&=-\frac{v^1}{(v^1-e^{v^2})^2}, \quad \theta_{0,-2}=\frac{2\theta_{2,1}}{(v^{1}-e^{v^2})^4}.
\end{align*}
Then the Principal Hierarchy of $M_\textrm{AL}$, as it is defined in \cite{liu2022generalized}, consists of the following Hamiltonian systems of hydrodynamic type:
\begin{align}
    \frac{\p v^{\al}}{\p t^{\beta,k}}&=\eta^{\al\ve}\frac{\p}{\p x}\left(\frac{\p \theta_{\beta,k+1}(v)}{\p v^{\ve}}\right)=\left\{v^\al(x), H^{[0]}_{\beta,k}\right\}_0,\quad \al,\beta=1,2,\,k\ge 0,\label{normal}\\
      \frac{\p v^{\al}}{\partial t^{0,\ell}}&=\eta^{\al\ve}\frac{\p}{\p x}\left(\frac{\p \theta_{0,\ell+1}(v)}{\p v^{\ve}}\right)=\left\{v^\al(x),H^{[0]}_{0,\ell}\right\}_0,\quad \al=1,2,\, \ell\in\mathbb{Z}.\label{0flow}
\end{align} 
Here the Poisson bracket $\{\cdot\,,\cdot\}_0$ is given by the Hamiltonian operator
\[\mathcal{P}_0^{[0]}=\begin{pmatrix}
0&\p_x\\
\p_x &0\end{pmatrix},\]
and the Hamiltonians are given by
\[H^{[0]}_{\beta,k}=\int \theta_{\beta,k+1}(v(x))\nd x,\quad H^{[0]}_{0,\ell}=\int \theta_{0,\ell+1}(v(x))\nd x.\]
These flows are mutually commutative, and they satisfy the following bihamitonian recursion relations:
\begin{align}
	&\left\{v^{\al}(x),H^{[0]}_{2,p-1}\right\}_1=(p+1)\left\{v^{\al}(x),H^{[0]}_{2,p}\right\}_0\label{zh-29},\\
	&\left\{v^{\al}(x),H^{[0]}_{1,p-1}\right\}_1=p\left\{v^{\al}(x),H^{[0]}_{2,p}\right\}_0+2\left\{v^{\al}(x),H^{[0]}_{2,p-1}\right\}_0,\\
	&\left\{v^{\al}(x),H^{[0]}_{0,p-1}\right\}_1=p\left\{v^{\al}(x),H^{[0]}_{0,p}\right\}_0+\{v^{\al}(x),H_{2,p-1}^{[0]}\}_0,\\
		 &\left\{v^{\al}(x),H^{[0]}_{0,-p-1}\right\}_1=-p\left\{v^{\al}(x),H^{[0]}_{0,-p}\right\}_0,\label{zh-30}
\end{align} 
where the Poisson bracket $\{\cdot\,,\cdot\}_1$ is given by the Hamiltonian operator
\[ \mathcal{P}_1^{[0]}=\begin{pmatrix}
(v^1e^{v^2})_x &(v^1+e^{v^2})\p_x\\
(v^1+e^{v^2})_x+(v^1+e^{v^2})\p_x &2\p_x
\end{pmatrix}.\]
Note that the flow $\frac{\p}{\p t^{0,0}}$ is given by the translation along the spatial variable $x$, i.e.
\[\frac{\p}{\p t^{0,0}}=\frac{\p}{\p x},\]
so in what follows we will identify the time variable $t^{0,0}$ with the spatial variable $x$.

The functions $\theta_{\al,k},\,\al=1,2, k\ge 0$ and $\theta_{0,\ell},\,\ell\in\mathbb{Z}$ can be represented in terms of the superpotential of $M_\textrm{AL}$ given by \eqref{zh-7}. To this end, we need to expand $\log(\lambda(p))$ in two ways. The first way is to assume $|v^1|<|p-e^{v^2}|<|e^{v^2}|$ and expand $\log(\lambda(p))$ as follows:
\begin{align}
\log^+(\lambda(p))&=v^2+\log\left(1+\frac{p-e^{v^2}}{e^{v^2}}\right)+\log\left(1+\frac{v^1}{p-e^{v^2}}\right)\notag\\
       &=v^2+\sum_{k\ge 1}(-1)^{k+1}\frac{1}{k}\left(\frac{p-e^{v^2}}{e^{v^2}}\right)^k+\sum_{k\ge 1}(-1)^{k+1}\frac{1}{k}\left(\frac{v^1}{p-e^{v^2}}\right)^k.
\end{align}
The second way is to assume $|e^{v^2}|<|p-e^{v^2}|<|v^1|$ and and expand $\log(\lambda)$ as follows:
\begin{align}
        \log^-(\lambda(p))&=\log v^1+\log\left(1+\frac{p-e^{v^2}}{v^1}\right)+\log\left(1+\frac{e^{v^2}}{p-e^{v^2}}\right)\notag\\
       & =\log v^1+\sum_{k\ge 1}(-1)^{k+1}\frac{1}{k}\left(\frac{p-e^{v^2}}{v^1}\right)^k+\sum_{k\ge 1}(-1)^{k+1}\frac{1}{k}\left(\frac{e^{v^2}}{p-e^{v^2}}\right)^k.
\end{align}
\begin{pro} 
We have the following formulae:
\begin{align}
\theta_{0,-k}&=(-1)^{k-1}k! \Res_{p=0}\lambda(p)^{-k}\frac{\nd p}{p},\label{zh-8}\\
\theta_{0,k}&=\frac1{k!}\Res_{p=e^{v^2}}\lambda^k(p)\left(\log^+(\lambda(p))-H_k\right)\frac{\nd p}{p},\label{zh-9}\\
\theta_{1,k}&=\frac1{k!}\Res_{p=e^{v^2}}\lambda^k(p)\left(\log^+(\lambda(p)+\log^-(\lambda(p))-2H_k\right)\frac{\nd p}{p},\label{zh-10}\\
\theta_{2,k}&=-\frac1{(k+1)!}\Res_{p=\infty}\lambda^{k+1}(p)\frac{\nd p}{p},\label{zh-11}
\end{align} 
here $k\ge 1$ and $H_k=1+\frac12+\dots+\frac1{k}$.  
\end{pro}
\begin{proof}
    From \cite{Brini2011Integrable} we know that the functions 
    $\theta_{2,k}$ and $\theta_{0,-k}$ can be represented in terms of hypergeometric functions, then it is easy to check by using these expressions that they satisfy the relations \eqref{re2},\eqref{quasi-2} and \eqref{re0-2}, \eqref{quasi0} respectively, thus we arrive at the validity of the formulae \eqref{zh-8}, \eqref{zh-11}.
    
    Let  $\Theta_{0,k}$ be the right hand side of \eqref{zh-9}
    \[\Theta_{0,k}=\frac1{k!}\Res_{p=e^{v^2}}\lambda^k(p)\left(\log^+(\lambda(p))-H_k\right)\frac{\nd p}{p}.\]
	To prove the validity of the formula \eqref{zh-9}, we need to show that $\Theta_{0,k}$ saitsifies 
		\begin{align*}
	\p_{\alpha}\p_{\beta}\Theta_{0,k+1}&=c_{\al\beta}^{\gamma}\p_{\gamma}\Theta_{0,k},\\
	E(\Theta_{0,k})&=k\Theta_{0,k}+\theta_{2,k-1},\\
\Theta_{1,1}&=\theta_{0,1}.
	\end{align*}
		Observe that taking derivatives and residues can be swapped in order and the above equations can be checked by direct calculation, here we use  the fact that the residue is zero at the analytic point, for example 
	\begin{align*}
	&\p_1\p_2\Theta_{0,k+1}-(e^{v^2}\p_1\Theta_{0,k}+\p_2\Theta_{0,k})\\
	=&\frac{1}{k!}\Res_{p=e^{v^2}}\left(\frac{e^{v^2}}{(p-e^{v^2})^2}\lambda^{k}(\log^{+}(\lambda(p))-H_k)+k\frac{pv^1e^{v^2}}{(p-e^{v^2})^3}\lambda^{k-1}(\log^{+}(\lambda(p))-H_{k-1})\right)\\
	&-\frac{1}{(k-1)!}\Res_{p=e^{v^2}}\left(\left(\frac{e^{v^2}}{p-e^{v^2}}+\frac{e^{2v^2}v^1}{(p-e^{v^2})^2}\right)\lambda^{k-1}(\log^{+}(\lambda(p))-H_{k-1})\right)\nd p\\
	=&\Res_{p=e^{v^2}}\frac{\p}{\p p}\left(-\frac{e^{v^2}\lambda^k(p)}{k!(p-e^{v^2})}(\log^{+}(\lambda(p))-H_k)\right)\nd p=0
	\end{align*}
	where $\p_{\alpha}=\frac{\p}{\p v^{\al}}$ and the validity of \eqref{zh-10} can be verified in a similar way.

  The proposition is proved.
\end{proof}
As it is shown in \cite{liu2022generalized}, the Principal Hierarchy  
\eqref{normal}, \eqref{0flow} of $M_\textrm{AL}$ possesses a tau cover
\begin{equation}\label{tc}
     \frac{\partial f^{[0]}}{\partial t^{\alpha,k}}=f^{[0]}_{\alpha,k}, \quad 
     \frac{\partial f_{\alpha,k}^{[0]}}{\partial t^{\beta,\ell}}=\Omega^{[0]}_{\alpha,k;\beta,\ell},\quad
     \frac{\partial v^{\gamma}}{\partial t^{\alpha,k}}=\eta^{\gamma\ve}\partial_x\Omega^{[0]}_{\ve,0;\alpha,k}, 
 \end{equation}
where $\gamma=1,2$, the indices $(\al,k), (\beta,\ell)$ belong to the set
\begin{equation}\label{zh-27}
\mathcal{I}=\{(\xi, q)\mid \xi=1,2,\,q\ge 0\}\cup \{(0,q)\mid q\in\mathbb{Z}\},
\end{equation}
and the functions $\Omega_{\al,k;\beta,\ell}$ are defined by
\begin{equation}\label{omega}
    \Omega^{[0]}_{\alpha,k;\beta,\ell}=\begin{cases}
    \sum\limits_{m=0}^\ell(-1)^{m}\langle \nabla \theta_{\alpha,k+1+m}, \nabla\theta_{\beta,\ell-m}\rangle, &\alpha,\beta=1,2; k,\ell\in\mathbb{N}\\
    \sum\limits_{m=0}^\ell(-1)^{m}\langle \nabla \theta_{\alpha,k+1+m}, \nabla\theta_{\beta,\ell-m}\rangle, &
    \alpha=0,k\in\mathbb{Z};\beta=1,2,\ell\in\mathbb{N}\\
    \sum\limits_{m=0}^{\ell-1}(-1)^{m}\langle \nabla \theta_{\alpha,k+1+m}, \nabla\theta_{\beta,\ell-m}\rangle+(-1)^{\ell}\theta_{0,k+\ell}, &
    \alpha,\beta=0,\,k\in\mathbb{Z},\ell\ge 0\\
     \sum\limits_{m=0}^{-\ell-1}(-1)^{m}\langle \nabla \theta_{\alpha,k+1-m}, \nabla\theta_{\beta,\ell+m}\rangle+(-1)^{\ell}\theta_{0,p+q}, &
    \alpha,\beta=0,\, p\in\mathbb{Z},\ell<0
    \end{cases}
\end{equation}
These functions satisfy the relations
\begin{align}
    \partial_x\Omega^{[0]}_{\alpha,k;\beta,\ell}&=\frac{\partial \theta_{\alpha,k}}{\partial t^{\beta,\ell}}=\frac{\partial \theta_{\beta,\ell}}{\partial t^{\alpha,k}},\\ \Omega^{[0]}_{\alpha,k;\beta,\ell}&=\Omega^{[0]}_{\beta,\ell;\alpha,k},\\ \Omega^{[0]}_{0,0;\alpha,k}&=\theta_{\alpha,k}.
\end{align}
For a solution $(f^{[0]}, f^{[0]}_{\al,k}, v^\gamma)$ of the tau cover, we call
\begin{equation}\label{4.1.5}
\tau^{[0]}=e^{f^{[0]}}
\end{equation}
the tau function for the solution $v^1(t), v^2(t)$ of the Principal Hierarchy \eqref{normal}, \eqref{0flow}.
 
The tau cover of the Principal Hierarchy possesses an infinite number of Virasoro symmetries, they can be represented in terms of the Virasoro operators of the form
\begin{equation}\label{virop}
L_m=a_m^{\al,p;\beta,q}\ve^2\frac{\p^2 }{\p t^{\al,p}\p t^{\beta,q}}+b^{\al,p}_{m;\beta,q} t^{\beta,q}\frac{\p }{\p t^{\al,p}}+\frac1{\ve^2} c_{m;\al,p;\beta,q} t^{\al,p} t^{\beta,q}+\kappa \delta_{m,0},\quad m\ge -1,
\end{equation}
where $a_m^{\al,p;\beta,q}=a_m^{\beta,q;\al,p}, c_{m;\al,p;\beta,q}=c_{m;\beta,q;\al,p}$.
These operators satisfy the Virasoro commutation relations
\begin{equation}
    [L_m,L_n]=(m-n)L_{m+n},\quad m, n\ge -1,
\end{equation}
and have the following explicit form:
\begin{align}
L_{-1}=&\sum_{k\ge 1}\left(t^{1,k}\frac{\partial}{\partial t^{1,k-1}}+t^{2,k}\frac{\partial}{\partial t^{2,k-1}}\right)+\sum\limits_{p\in\mathbb{Z}}t^{0,p}\frac{\partial}{\partial t^{0,p-1}}+\frac1{\ve^2}t^{1,0}t^{2,0},\label{vir-1}\\
L_0=&\sum_{k\ge 1}k\left(t^{1,k}\frac{\partial}{\partial t^{1,k}}+t^{2,k-1}\frac{\partial}{\partial t^{2,k-1}}\right)+\sum_{p\in\mathbb{Z}}pt^{0,p}\frac{\partial}{\partial t^{0,p}}\notag\\
&+\sum_{k\ge 1}\left(2t^{1,k}+t^{0,k}\right)\frac{\partial}{\partial t^{2,k-1}}+\frac1{\ve^2}(t^{1,0})^2+\frac1{\ve^2}\sum_{k\ge 0}(-1)^kt^{0,-k}t^{1,k},\label{vir0}\\
L_m=&\sum_{k\ge 1}\beta_{m}(k)\left(t^{1,k}\frac{\partial}{\partial t^{1,k+m}}+t^{2,k-1}\frac{\partial}{\partial t^{2,k+m-1}}+t^{0,k}\frac{\partial}{\partial t^{0,k+m}}\right.\notag\\
&\left.+(-1)^{m+1}t^{0,-k-m}\frac{\partial}{\partial t^{0,-k}}\right)
		+\sum_{k\ge 0}\alpha_m(k)\left(2t^{1,k}+t^{0,k}\right)\frac{\partial}{\partial t^{2,k+m-1}}\notag\\
		&+\sum_{k=1-m}^{-1}\alpha_m(k)\left((-1)^k\ve^2 \frac{\partial^2}{\partial t^{2,k+m-1}\partial t^{2,-k-1}}
		+t^{0,k}\frac{\partial}{\partial t^{2,k+m-1}}\right)\notag\\
			&+\frac1{\ve^2}\sum_{k\ge 0}\alpha_m(k)(-1)^{k+m}t^{0,-k-m}t^{1,k},\quad m\ge 1.\label{vir-m}
    \end{align}
Here the constants $\al_{m}(k)$ and $\beta_{m}(k)$ are defined by
    \[\beta_{m}(k)=\frac{(m+k)!}{(k-1)!},\qquad\alpha_m(k)=\begin{cases}
   \frac{(m+k)!}{(k-1)!}\sum\limits_{j=k}^{m+k}\frac{1}{j},  &k>0,  \\
    m!, &k=0,\\
    (-1)^{k}(-k)!(k+m)!,&-m<k<0.
\end{cases}\]
The actions of these Virasoro symmetries on the tau cover of the Principal Hierarchy of $M_\textrm{AL}$ are given by
\begin{align}
\frac{\p f^{[0]}}{\p s_m}&
=a_m^{\al,p;\beta,q}f^{[0]}_{\al,p}f^{[0]}_{\beta,q}+b^{\al,p}_{m;\beta,q} t^{\beta,q}f^{[0]}_{\al,p}+c_{m;\al,p;\beta,q} t^{\al,p} t^{\beta,q},\label{zh-34}\\
\frac{\p f^{[0]}_{\al,k}}{\p s_m}&=\frac{\p}{\p t^{\al,k}}\left(\frac{\p f^{[0]}}{\p s_m}\right),\quad \frac{\p v^\gamma}{\p s_m}=\eta^{\gamma\zeta}\frac{\p}{\p x}\left(\frac{\p f^{[0]}_{\zeta,0}}{\p s_m}\right),\quad m\ge -1.\label{zh-35}
\end{align}
The flows $\frac{\p}{\p s_m}$ commute with the flows of the tau cover, i.e.
\begin{equation}\label{vir}
 \left[\frac{\p}{\p t^{\al,k}},\frac{\p}{\p s_m}\right]f^{[0]}=\left[\frac{\p}{\p t^{\beta,\ell}},\frac{\p}{\p s_m}\right]f^{[0]}_{\alpha,k}=\left[\frac{\p}{\p t^{\beta,\ell}},\frac{\p}{\p s_m}\right]v^{\gamma}=0,    
\end{equation}
for any $(\al,k)$, $(\beta,\ell)\in\mathcal{I}$, $\gamma=1,2$ and $m\ge-1$. 

\section{Deformation of the Principal Hierarchy of $M_\text{AL}$}

In this section, we consider a tau-symmetric deformation of the Principal Hierarchy \eqref{normal}, \eqref{0flow} of $M_\textrm{AL}$. We first rewrite the positive and negative flows \eqref{lax1}, \eqref{lax2} of the Ablowitz-Ladik hierarchy in terms of the new time variables $t^{2,p}, t^{0,-p}$ as follows:
\begin{align}
\frac{\partial L}{\partial t^{2,p}}&=\frac{1}{\ve(p+1)!}[(L^{p+1})_{+},L],\quad p\ge 0\label{lax1}\\
    \frac{\partial L}{\partial t^{0,-p}}&=\frac{(-1)^{p-1} (p-1)!}{\ve}[(M^p)_{-},L], \quad p\ge 1\label{lax2}
\end{align}
where the operators $L, M$ are given by \eqref{laxop1}, \eqref{laxop2}. 
Let us introduce the unknown functions
\begin{equation}\label{zh-37}
w^1=Q-P,\quad w^2=\log Q.
\end{equation}
Then the flows \eqref{lax1}, \eqref{lax2} can be represented as bihamiltonian systems of the form \cite{li2022tri, oevel1989mastersymmetries}
\begin{align}
\frac{\p}{\p t^{2,p}}\begin{pmatrix} w^1\\ w^2\end{pmatrix}
&=\mathcal{P}_0\begin{pmatrix}\frac{\delta{H_{2,p}}}{\delta w^1}\\[5pt]
\frac{\delta{H_{2,p}}}{\delta w^2}\end{pmatrix}=\frac1{p+1}\mathcal{P}_1\begin{pmatrix}\frac{\delta{H_{2,p-1}}}{\delta w^1}\\[5pt]
\frac{\delta{H_{2,p-1}}}{\delta w^2}\end{pmatrix},\quad p\ge 0,\\[3pt]
\frac{\p}{\p t^{0,-p}}\begin{pmatrix} w^1\\ w^2\end{pmatrix}
&=\mathcal{P}_0\begin{pmatrix}\frac{\delta{H_{0,-p}}}{\delta w^1}\\[5pt]
\frac{\delta{H_{0,-p}}}{\delta w^2}\end{pmatrix}=-\frac1{k}\mathcal{P}_1\begin{pmatrix}\frac{\delta{H_{0,-p-1}}}{\delta w^1}\\[5pt]
\frac{\delta{H_{0,-p-1}}}{\delta w^2}\end{pmatrix},\quad p\ge 1,
\end{align}
where Hamiltonian operators are given by
	\begin{align}\label{bho}
	    \mathcal{P}_0&=\ve^{-1}
	\begin{pmatrix}0&\Lambda-1
\\1-\Lambda^{-1}& 0\end{pmatrix},	\\
\mathcal{P}_1&=\ve^{-1}\begin{pmatrix}
    -e^{w^2}\Lambda^{-1} w^1+w^1\Lambda e^{w^2}&w^1(\Lambda-1)+e^{w^2}(1-\Lambda^{-1})\\(1-\Lambda^{-1})w^1+(\Lambda-1)e^{w^{2}}&\Lambda-\Lambda^{-1}
\end{pmatrix},\label{bho2}
	\end{align}
and the Hamiltonians 
\begin{equation}\label{zh-16}
H_{2,p}=\int h_{2,p+1}\nd x,\quad H_{0,-q}=\int h_{0,-q+1}\nd x,\quad p\ge -1,\,q\ge 1
\end{equation}
have the densities
\begin{align}
    h_{2,p}&=\frac{1}{(p+1)!}\Res L^{p+1},\quad p\ge0,\\
   h_{0,-q}&=\begin{cases}
     \frac{\ve \p_x}{1-\Lambda^{-1}} \left(w^2-\log(e^{w^2}-w^1)\right),\quad q=0,\\[5pt]
       (-1)^q(q-1)!\Res M^{q},\quad q\ge1.\label{zh-19}
    \end{cases}
\end{align}
Here the residue of an operator $Y=\sum\limits_{n\in\mathbb{Z}} Y^{n}\Lambda^{n}$ is defined by $\Res Y=Y_0$. Then the following tau-symmetry conditions hold true:
\begin{align}
    \frac{\p h_{2,p}}{\partial t^{2,q}}&=\frac{1}{(p+1)!}\Res\frac{\partial L^{p+1}}{\partial t^{2,q}}=\frac{1}{(p+1)!(q+1)!}\Res([(L^{q+1})_{+},L^{p+1}])\nonumber\\
   & =\frac{1}{(p+1)!(q+1)!}\Res([(L^{p+1})_{+},L^{q+1}])=\frac{\partial h_{2,q}}{\partial t^{2,p}},\quad p, q\ge 0;\\
  \frac{\partial h_{2,p}}{\partial t^{0,-q}}&=\frac{1}{(p+1)!}\Res\frac{\partial L^{p+1}}{\partial t^{0,-q}}=
  \frac{(-1)^{q-1}(q-1)!}{(p+1)!}\Res([(M^{q})_{-},L^{p+1}])\nonumber\\
    &=\frac{(-1)^{q-1}(q-1)!}{(p+1)!}\Res([(M^{q})_{-},(L^{p+1})_+])=
    \frac{\p h_{0,-q}}{\partial t^{2,p}},\quad p\ge 0, q\ge1;\\
 \frac{\partial h_{0,-p}}{\partial t^{0,-q}}&=(-1)^p(p-1)!\Res\frac{\partial M^{p}}{\partial t^{0,-q}}=(-1)^{q+p-1}(p-1)!(q-1)!\Res([(M^{q})_{-},M^{p}])\nonumber\\
    &=(-1)^{q+p-1}(p-1)!(q-1)!\Res([(M^{p})_{-},M^{q}])=\frac{\partial h_{0,-q}}{\partial t^{0,-p}},\quad p, q\ge1.
\end{align}
 It is easy to check that the leading terms of $h_{2,p}$, $h_{0,-p}\, (k\ge 0)$ coincide with the Hamiltonian densities $\theta_{2,p}$, $\theta_{0,-p}$ of the Principal Hierarchy, so the positive and negative flows 
\eqref{lax1}, \eqref{lax2} are tau-symmetric deformations of the flows $\frac{\p}{t^{2,p}}\, (p\ge 0)$ and $\frac{\p}{t^{0,-p}}\, (p\ge 1)$ of the Principal Hierarchy \eqref{normal}, \eqref{0flow} of $M_\textrm{AL}$.  
 
Now let us proceed to show the existence of tau-symmetric deformations of the logarithm flows $\frac{\p}{t^{1,k}},  \frac{\p}{t^{0,k}}\,(k\ge 0)$ of the Principal Hierarchy of $M_\textrm{AL}$. To this end, we need to use the canonical coordinates $u^1, u^{2}$ of the generalized Frobenius manifold $M_\text{AL}$ that are given in \cite{brini2012local}:
\begin{equation}
    u^{1}=e^{v^2}+v^1+2\sqrt{v^{1}e^{v^2}},\quad
    u^2=e^{v^2}+v^1-2\sqrt{v^{1}e^{v^2}}.
\end{equation}
The flat coordinates $v^1,v^2$ can be represented in terms of the canonical coordinates as follows:
\begin{equation}
v^{1}=\left(\frac{\sqrt{u^1}-\sqrt{u^2}}{2}\right)^2,\quad v^2=2\log\left(\frac{\sqrt{u^1}+\sqrt{u^2}}{2}\right).
\end{equation}
In the canonical coordinates, the flows of in the Principal Hierarchy of $M_\text{AL}$ can be represented as
\begin{equation}
    \frac{\p u^{i}}{\p t^{\beta,p}}=\frac{\p u^{i}}{\p v^{\gamma}}c^{\gamma\zeta}_{\xi}\frac{\p \theta_{\beta,q}}{\p v^{\zeta}}v^{\xi}_x=\eta^{\gamma\zeta}\frac{\p u^{i}}{\p v^{\gamma}}\frac{\p u^{i}}{\p v^{\xi}}\frac{\p \theta_{\beta,q}}{\p v^{\zeta}}v^{\xi}_x=\eta^{\gamma\zeta}\frac{\p u^{i}}{\p v^{\gamma}}\frac{\p \theta_{\beta,q}}{\p v^{\zeta}}u^{i}_x=V_{\beta,q}^i(u)u^{i}_x,\quad i=1,2.
\end{equation}
Here and in what follows we use Latin indices for canonical coordinates, and we do not assume summations over repeated Latin indices unless explicitly indicated. 

Let us denote by $\mathcal{A}$ the space of differential polynomials on $M_{\textrm{AL}}$, and by $\mathcal{A}^{(N)}$ the subspace of $\mathcal{A}$ consists of differential polynomials which do not depend on
\[u^{i,s}=\p_x^s u^i,\quad i=1,2\]
 for $s>N$. We also denote by $\mathcal{A}_k$, $\mathcal{A}^{(N)}_k$ the subspaces of $\mathcal{A}$ and $\mathcal{A}^{(N)}$ consisted of elements of differential degree $k$.

We need the following lemma and its corollary to prove the existence of tau-symmetric deformations of the logarithm flows.
\begin{lem}\label{1}
	 Let $g_q(u,u_x,\dots,u^{(N+1)})\in\mathcal{A}^{(N+1)}_{\ge 1}$ ( $q\ge 1$) satisfy the conditions
	 \begin{equation}\label{sym}
	 \frac{\partial g_q}{\partial t^r}=\frac{\partial g_r}{\partial t^q},\quad r, q\ge 1,
	 \end{equation}
with respect to the flows $\frac{\p}{\p t_q}$ of the form
	 \begin{equation}
	 	\frac{\p u^i}{\p t^{q}}=V^{i}_q(u) u_{x}^i,\quad i=1,2,\, q\ge 1,
	 \end{equation}
which are assumed to commute with each other, and the nonzero smooth functions $V^{i}_q(u)$ satisfy the conditions
	\begin{align}
&V_q^i(u)V_r^j(u)-V_r^i(u)V_q^j(u)\neq 0,\quad i\neq j, \, q\neq r,\label{zh-12}\\
&\frac{\p V^i_q}{\p u^i}V^i_r-\frac{\p V^i_r}{\p u^i}V^i_q\neq 0,\quad q\ne r.\label{zh-14}
\end{align}
Then there exists $f(u,u_x,\dots,u^{(N)})\in \mathcal{A}^{(N)}$ such that
	\[\frac{\partial f}{\partial t^q}=g_q,\quad q\ge 1.\]
		\end{lem}
\begin{proof}
We first show the existence of $H^1, H^2, B_q\in\mathcal{A}^{(N)}$ such that $g_q\, (q\ge 1)$ has the following expression: 
\begin{equation}\label{gq}
g_q=V^1_q H^1(u,\dots,u^{(N)})u^{1,N+1}+V_q^2 H^2(u,\dots,u^{(N)})u^{2,N+1}+B_q(u,\dots,u^{(N)}).
\end{equation}
Let us represent the conditions \eqref{sym} in the form
\begin{align}
      \sum\limits_{i=1}^2\left(\frac{\partial g_q}{\partial u^{i,N+1}}\partial_x^{N+1}(V_r^{i}u^{i,1})+ \frac{\partial g_q}{\partial u^{i,N}}\partial_x^{N}(V_r^{i}u^{i,1})+\dots\right)\notag\\
      =\sum\limits_{i=1}^2\left(\frac{\partial g_r}{\partial u^{i,N+1}}\partial_x^{N+1}(V_q^{i}u^{i,1})+ \frac{\partial g_r}{\partial u^{i,N}}\partial_x^{N}(V_q^{i}u^{i,1})+\dots\right),\label{11}
\end{align}
By comparing the coefficients of $u^{i,N+2}$ on both sides of the above equations we obtain
\begin{equation}\label{N+2}
   \frac{\partial g_q}{\partial u^{i,N+1}}V^{i}_r=\frac{\partial g_r}{\partial u^{i,N+1}}V^{i}_q,\quad \text{for}\ i=1,2. 
\end{equation}
Thus by using the condition \eqref{zh-12} we arrive at
\[\frac{\partial^2 g_q}{\partial u^{i,N+1}\partial u^{j,N+1}}=0,\quad i\ne j,\]
and we can represent $g_q$ in the form
\[
    g_q=V_q^1H_1(u,\dots,u^{(N)},u^{1,N+1})+V_q^2H_2(u\dots,u^{(N)},u^{2,N+1})+A_q(u,\dots,u^{(N)}),
\]
here $H_1, H_2, A_q$ are certain differential polynomials.

Let us proceed to show that $H_{1}$ depends on $u^{1,N+1}$ linearly. To this end we define a lexicographical order on the set of monomials as follows:
\[
 (u^{1,N+1})^{\beta_{N+1}}(u^{1,N})^{\beta_{N}}\dots(u^{1,1})^{\beta_1}\preceq(u^{1,N+1})^{\gamma_{N+1}}(u^{1,N})^{\gamma_{N}}\dots(u^{1,1})^{\gamma_1}\]
 if and only if either $\beta_{N+1}<\gamma_{N+1}$ or $\beta_{N+1}=\gamma_{N+1}$ and 
 \[(u^{1,N})^{\beta_{N}}(u^{1,N-1})^{\beta_{N-1}}\dots(u^{1,1})^{\beta_1}\preceq(u^{1,N})^{\gamma_{N}}(u^{1.N-1})^{\gamma_{N-1}}\dots(u^{1,1})^{\gamma_1}.\]
Then $H_1$ can be written as the sum of the highest order monomial and lower order ones as follows:
 \[H_1=c(u^{1,N+1})^{\beta_{N+1}}(u^{1,N})^{\beta_N}\dots(u^{1,1})^{\beta_1}+R,\]
 where $c\in\mathcal{A}^{(N)}$ do not depend on $u^{1,1},\dots,u^{1,N}$, and $c\ne 0$.
 If $\beta_{N+1}>1$, then by applying the derivatives
\[\frac{\p}{\p u^{2,N+1}}\left(\frac{\p}{\p u^{1,N+1}}\right)^{\beta_{N+1}}\]
on both sides of \eqref{11} obtain 
 \begin{equation}
     (V^2_q V^1_r-V^1_q V^2_r)\frac{\partial c}{\partial u^{2,N}}=0,
 \end{equation}
thus we have $\frac{\partial c}{\partial u^{2,N}}=0$. Similarly, by applying 
the derivatives
\[\frac{\p}{\p u^{2,N}}\left(\frac{\p}{\p u^{1,N}}\right)^{\beta_{N}}\left(\frac{\p}{\p u^{1,N+1}}\right)^{\beta_{N+1}}\]
on both sides of \eqref{11} we know that 
$\frac{\partial c}{\partial u^{2,N-1}}=0$. In this way we arrive at $c\in \mathcal{A}_0$. Now from the coefficients of $(u^{1,N+1})^{\beta_{N+1}}(u^{1,N})^{\beta_N}\dots(u^{1,1})^{\beta_1+1}$ on both sides of \eqref{11} we obtain
 \begin{equation}
   c\left(\beta_1+3\beta_2\dots+(N+2)\beta_{N+1}-1\right)\left(\frac{\partial V^1_r}{\partial u^1}V^1_q-\frac{\partial V^1_q}{\partial u^1}V^1_r\right)=0.
 \end{equation}
Then from the condition \eqref{zh-14} it follows that
$c=0$, which contradicts our assumption that 
$c\ne 0$. Thus we arrive at the fact that $\beta_{N+1}=1$, i.e, $H_1$ depends linearly on $u^{1,N+1}$. We can show in the same way that $H_2$
depends on $u^{2,N+1}$ also linearly.

Now we can assume that $g_q, g_r$ have the form \eqref{gq}. By taking the derivative
\[\frac{\partial ^2}{\partial u^{1,N+1}\partial u^{2,N+1}}\]
on both sides of \eqref{11}
we obtain
 \[(V^2_q V^1_r-V^1_q V^2_r)\left(\frac{\partial H^1}{\partial u^{2,N}}-\frac{\partial H^2}{\partial u^{1,N}}\right)=0,\]
so we have
\[
    \frac{\partial H^1}{\partial u^{2,N}}=\frac{\partial H^2}{\partial u^{1,N}}.
\]
Hence there exists $\tilde{f}\in\mathcal{A}^{(N)}$ such that
\[\frac{\partial \tilde{f}}{\partial u^{i,N}}=H^{i},\quad i=1,2.\]
Let $\tilde{g}_q=g_q-\frac{\partial \tilde{f}}{\partial t^q}$, then we have 
\[\frac{\partial\tilde{g}_q}{\partial t^r}=\frac{\partial\tilde{g}_r}{\partial t^q},\]
for any $q,r\ge 1$, and $\tilde{g}_q\in\mathcal{A}^{(N)}$. Thus by induction on $N$, we can reduce the proof of the lemma to the case when 
$g_q\in\mathcal{A}^{(1)}_{\ge  1}$, and to prove the existence of $f\in \mathcal{A}^{(0)}$ such that 
\[\frac{\partial f}{\partial t^q}=g_q.\]
In fact, by repeating the above discussion and by using the commutativity of the flows $\frac{\p}{\p t^q}\,(q\ge 1)$, we can show that $g_q$ can be represented in the form
\[g_q=V^1_q H^1(u)u^{1,1}+V^2_q H^2(u)u^{2,1},\,\text{with}\
  \frac{\partial H^1}{\partial u^2}=\frac{\partial H^2}{\partial u^1}.\]
Hence the lemma is proved.
\end{proof}
	\begin{rem}
It can be checked that the flows of the Principal Hierarchy of $M_{\text{AL}}$ satisfy the conditions \eqref{zh-12}, \eqref{zh-14} of Lemma \ref{1}.
\end{rem}
\begin{cor}\label{w-2}
Assume that the differential polynomials 
\begin{equation}
	g_q=g_{q,0}+\ve g_{q,1}+\dots,\quad g_{q,k}\in\mathcal{A}_{k+1},\quad k\ge 0,\,q\ge 1
\end{equation}
satisfy the condition
\begin{equation}
	\frac{\p g_q}{\p t^{r}}=\frac{\p g_r}{\p t^{q}},\quad q, r\ge 1
\end{equation}
with respect to the commuting flows of the form
\begin{equation}\label{zh-15}
\frac{\p u^i}{\p t^q}=V^i_q(u) u^{i,1}+\ve U^i_1+\ve^2 U^i_2+\dots,
\quad U^i_k\in\mathcal{A}_{k+1},\, k\ge 1,
\end{equation}
here the leading terms of the flows satisfy the conditions given in Lemma \ref{1}. Then there exists $f\in\mathcal{A}$ such that 
\begin{equation}
	\frac{\p f}{\p t^{q}}=g_q,\quad q\ge 1.
\end{equation} 
\end{cor}
\begin{proof}
	We denote the flows given by the leading terms of \eqref{zh-15} by
	\[\frac{\p u^i}{\p t^q_0}=V^i_q(u) u^{i,1},\quad i=1,2,\, q\ge 1.\]
Then it follows from Lemma \ref{1} the existence of $f_0\in\mathcal{A}_0$ such that
	\begin{equation}
		\frac{\p f_0}{\p t^{q}_0}=g_{q,0},\quad q\ge 1.
	\end{equation} 
	Denote
	\begin{equation}
		g_q^{[1]}=g_q-\frac{\p f_0}{\p t^{q}}=\ve g_{q,1}^{[1]}+\ve^2 g_{q,2}^{[1]}+\dots,\quad q\ge 1,
	\end{equation}
	then we
	\begin{equation}
		\frac{\p g_q^{[1]}}{\p t^{r}}=\frac{\p g_r^{[1]}}{\p t^{q}},\quad q, r\ge 1.
	\end{equation}
Thus by using Lemma \ref{1} again we can find $f_1\in\mathcal{A}_1$ such that 
	\begin{equation}
		\frac{\p f_1}{\p t^{q}_0}=g_{q,1}^{[1]},\quad q\ge 1.
	\end{equation}
By continuing this procedure in a recursive way, we can find $f_k\in\mathcal{A}_k, k\ge 2$ such that
\[g_q^{[k+1]}=g_q^{[k]}-\ve^k \frac{\p f_k}{\p t^q}=\ve^{k+1} g_{q,k+1}^{[k+1]}+\ve^{k+2} g_{q,k+2}^{[k+1]}+\dots.\]
Let $f=f_0+\ve f_1+\ve^2 f_2+\dots$, then we have
\[\frac{\p f}{\p t^q}=g_q,\quad q\ge 1.\]
The corollary is proved.
\end{proof}

We know from Theorem 6.2 of \cite{dubrovin2018bihamiltonian} that there exist bihamiltonian conserved quantities 
\[H_{0,p}=\int h_{0,p+1}(w, w_x,\dots)\nd x, \quad H_{1,p}=\int h_{1,p+1}(w,w_x,\dots)\nd x,\quad p\ge -1\]
of the bihamiltonian structure $(\mathcal{P}_0,\mathcal{P}_1)$
with densities of the form
	\[
	h_{\al,p}(w,w_x,\dots)=\theta_{\al,p}(w)+\ve h_{\al,p}^{[1]}(w,w_x)+\ve^2 h_{\al,p}^{[2]}(w,w_x,w_{xx})+\dots,\quad \al=0,1,\,p\ge 0,\]
here $h_{\al,p}^{[k]}\in\mathcal{A}_k$.
	These Hamiltonians together with the ones defined in \eqref{zh-16} are in involution with respect to the bihamiltonian structure, i.e.,
	\begin{equation}\label{w-0}
	\{H_{\al,p},H_{\beta,q}\}_{\mathcal{P}_0}=\{H_{\al,p},H_{\beta,q}\}_{\mathcal{P}_1}=0,\quad (\al,p),\,(\beta,q)\in\mathcal{I},
	\end{equation}
	they yield the following deformation of the Principal Hierarchy of $M_\textrm{AL}$:
	\begin{equation}\label{w-1}
	\frac{\p w^{\al}}{\p t^{\beta,q}}=\mathcal{P}_0^{\al\gamma}\frac{\delta H_{\beta,p}}{\delta w^{\gamma}},\quad \al=1, 2,\, (\beta,q)\in\mathcal{I},
	\end{equation}
which consists of mutually commutative flows. Since the flow given by the translation along the spatial variable $x$ is a symmetry 
for the flows of \eqref{w-1}, we know from Corollary A.4 of \cite{dubrovin2018bihamiltonian} that the $\frac{\p}{\p t^{0,0}}$-flow of \eqref{w-1} is given by
\[\frac{\p w^\al}{\p t^{0,0}}=w^\al_x,\quad \al=1,2.\]
Thus we will identify the spatial variable $x$ with $t^{0,0}$ in what follows.
We call the bihamiltonian integrable hierarchy \eqref{w-1} the extended Ablowitz-Ladik hierarchy.

The next proposition shows that we can choose the densities of the above Hamiltonians such that they are tau-symmetric.

\begin{pro}\label{leading0}
The densities $h_{0,p}(w, w_x,\dots), h_{1,p}(w, w_x,\dots)$ of the Hamiltonians $H_{0,p-1}$ and $H_{1,p-1}\,(p\ge 0)$ can be chosen to satisfy the tau-symmetry conditions
\begin{align}
    \frac{\partial h_{0,p}}{\partial t^{0,m}}&=\frac{\partial h_{0,m}}{\partial t^{0,p}},\quad \frac{\partial h_{0,p}}{\partial t^{2,q}}=\frac{\partial h_{2,q}}{\partial t^{0,p}},\quad p, q\ge 0, m\in\mathbb{Z},\\
    \frac{\partial h_{1,p}}{\partial t^{2,q}}&=\frac{\partial h_{2,q}}{\partial t^{1,p}},\quad \frac{\partial h_{1,p}}{\partial t^{0,m}}=\frac{\partial h_{0,m}}{\partial t^{1,p}},\quad\frac{\partial h_{1,p}}{\partial t^{1,q}}=\frac{\partial h_{1,q}}{\partial t^{1,p}},\quad p, q\ge 0, m\in\mathbb{Z}.\label{zh-18}
\end{align}
\end{pro}
\begin{proof}
Let us first show that the function $h_{0,0}$ defined in \eqref{zh-19}
satisfies the condition
\[\frac{\p h_{0,0}}{\p t^{2,q}}=\frac{\p h_{2,q}}{\p t^{0,0}},\quad q\ge 0.\]
From \cite{li2022tri} we know that the variational derivatives of the Hamiltonians $H_{2,q}$ with respect to the unkown functions $P, Q$ have the expressions
\[
\frac{\delta H_{2,q}}{\delta P}=-\frac{1}{(q+1)!}\mathrm{Res}(L^{q+1}B^{-1}),\quad  \frac{\delta H_{2,q}}{\delta Q}=\frac{1}{(q+1)!}\mathrm{Res}(\Lambda ^{-1}L^{q+2}B^{-1}),\]
and the flows $\frac{\p}{\p t^{2,q}}$ can be represented in the form
\begin{align*}
\epsilon\frac{\partial P}{\partial t^{2,q}}&=\frac{1}{q+1}P(\Lambda-1)Q\frac{\delta H_{2,q-1}}{\delta Q}=\frac{1}{(q+1)!}P(\Lambda-1)Q\Res (\Lambda^{-1}L^{q+1}B^{-1}),\\
\epsilon\frac{\partial Q}{\partial t^{2,q}}&=Q(\Lambda^{-1}-1)\frac{\delta H_{2,q}}{\delta P}=\frac{1}{(q+1)!}Q(1-\Lambda^{-1})\Res (L^{q+1}B^{-1}).
\end{align*}
Hence we have
\begin{align*}
\frac{\p h_{0,0}}{\p t^{2,q}}&=\frac{1}{(q+1)!}\p_x (\Res( L^{q+1}B^{-1})-\Lambda Q\Res (\Lambda^{-1}L^{q+1}B^{-1}))\\
    &=\frac{1}{(q+1)!}\p_x (\Res (L^{q+1}B^{-1})-\Res (L^{q+1}B^{-1}Q\Lambda^{-1}))\\
    &=\frac{1}{(q+1)!}\p_x \Res L^{q+1}=\frac{\p h_{2,q}}{\p t^{0,0}}.
\end{align*}
Here $B=1-Q\Lambda^{-1}$.

For any given $k\ge 1$, we denote 
\[g_{k,q}=\frac{\p h_{2,q}}{\p t^{0,k}},\quad q\ge 0.\] 
Then by using the tau-symmetry property of the Hamiltonian densities 
$h_{2, q}\, (q\ge 0)$ we obtain 
\begin{align*}
  \frac{\p g_{k,q}}{\p t^{2,r}}&=\frac{\partial}{\partial t^{2,r}}\left(\frac{\partial h_{2,q}}{\partial t^{0,k}}\right)=\frac{\partial}{\partial t^{0,k}}\left(\frac{\partial h_{2,q}}{\partial t^{2,r}}\right )=\frac{\partial}{\partial t^{0,k}}\left(\frac{\partial h_{2,r}}{\partial t^{2,q}}\right)\\
  &=\frac{\partial}{\partial t^{2,q}}\left(\frac{\partial h_{2,r}}{\partial t^{0,k}}\right)=\frac{\p g_{k,r}}{\p t^{2,q}},\quad q, r\ge 0.
\end{align*}
From the tau-symmetry property of the densities $\theta_{\al,p}$ of the Hamiltonians of the Principal Hierarchy of $M_{\textrm{AL}}$ and Corollary \ref{w-2}, it follows the existence of differential polynomials $\tilde{h}_{0,k}$ of the form
\[\tilde{h}_{0,k}(w,w_x,\dots)=\theta_{0,k}(w)+\ve \tilde{h}^{[1]}_{0,k}(w, w_x)+
\ve^2 \tilde{h}^{[2]}_{0,k}(w, w_x,w_{xx})+\dots,\quad \tilde{h}^{[k]}_{0,k}\in\mathcal{A}_k\]
such that 
\begin{equation}\label{w-3}
     \frac{\partial \tilde{h}_{0,k}}{\partial t^{2,q}}=g_{k,q}=\frac{\partial h_{2,q}}{\partial t^{0,k}}.
\end{equation}
Let 
\[\tilde{H}_{0,k-1}=\int \tilde{h}_{0,k}\dd x,\quad k\ge 1,\]
 then from \eqref{w-0} and \eqref{w-3} it follows that $\tilde{H}_{0,k-1}$ is a conserved quantity of the flows $\frac{\p}{\p t^{2,q}}$. Since $H_{0,k-1}$ is also a conserved quantity of the flows $\frac{\p}{\p t^{2,q}}$ and it shares the same leading term with that of $\tilde{H}_{0,k-1}$, we conclude
from Corollary A.4 of \cite{dubrovin2018bihamiltonian} that 
\[\tilde{H}_{0,k-1}=H_{0,k-1}.\]
Thus we can choose the densities $h_{0,k}$ of the Hamiltonian $H_{0,k-1}$ so that
\[ h_{0,k}=\tilde{h}_{0,k},\quad k\ge 1.\]
Now for any $p, q\ge 0$ we have 
\begin{align*}
\frac{\partial}{\partial t^{2,k}}\left(\frac{\partial h_{0,p}}{\partial t^{0,-q}}\right)&=\frac{\partial}{\partial t^{0,-q}}\left(\frac{\partial h_{2,k}}{\partial t^{0,p}}\right)=\frac{\partial}{\partial t^{0,p}}\left(\frac{\partial h_{2,k}}{\partial t^{0,-q}}\right)=\frac{\partial}{\partial t^{2,k}}\left(\frac{\partial h_{0,-q}}{\partial t^{0,p}}\right) \\
    \frac{\partial}{\partial t^{2,k}}\left(\frac{\partial h_{0,p}}{\partial t^{0,q}}\right)&=\frac{\partial}{\partial t^{0,q}}\left(\frac{\partial h_{2,k}}{\partial t^{0,p}}\right)=\frac{\partial}{\partial t^{0,p}}\left(\frac{\partial h_{2,k}}{\partial t^{0,q}}\right)=\frac{\partial}{\partial t^{2,k}}\left(\frac{\partial h_{0,q}}{\partial t^{0,p}}\right).
\end{align*}
Hence 
\begin{equation}
    \frac{\partial}{\partial t^{2,k}}\left(\frac{\partial h_{0,p}}{\partial t^{0,q}}-\frac{\partial h_{0,p}}{\partial t^{0,q}}\right)= \frac{\partial}{\partial t^{2,k}}\left(\frac{\partial h_{0,p}}{\partial t^{0,-q}}-\frac{\partial h_{0,-q}}{\partial t^{0,p}}\right)=0,
\end{equation}
 and we arrive at
\[\frac{\partial h_{0,p}}{\partial t^{0,q}}-\frac{\partial h_{0,q}}{\partial t^{0,p}}=0,\quad \frac{\partial h_{0,p}}{\partial t^{0,-q}}-\frac{\partial h_{0,-q}}{\partial t^{0,p}}=0,\quad p, q\ge 0.\]

 Similarly, we can choose the the densities 
 \[h_{1,p}(w,w_x,\dots)=\theta_{1,n}(w)+\ve h^{[1]}_{1,p}(w,w_x)+\ve^2 h^{[2]}_{1,p}(w,w_x,w_{xx})+\dots,\quad h_{1,p}^{[k]}\in\mathcal{A}_k \]
 of the Hamiltonians $H_{1,p-1}\,(p\ge 0)$ 
 satisfying the tau-symmetry condition \eqref{zh-18}.
The proposition is proved.
\end{proof}

Thus we can fix the densities $h_{\al,p}(w, w_x,\dots)$ of the Hamiltonians $H_{\al,p-1}$ of the deformed Principal Hierarchy \eqref{w-1} which satisfy the tau-symmetry condition
\[\frac{\p h_{\al,p}}{\p t^{\beta,q}}=\frac{\p h_{\beta,q}}{\p t^{\al,p}},\quad (\al, p), (\beta, q)\in 
\mathcal{I},\]
and we can define a set of differential polynomials 
$\{\Omega_{\alpha,p;\beta,q}(w,w_x,\dots)\}$, called a tau structure of the deformed Principal Hierarchy \eqref{w-1}, such that
       \[\Omega_{\alpha,p;\beta,q}(w,w_x,\dots)=\Omega_{\alpha,p;\beta,q}^{[0]}(w)+\ve \Omega_{\alpha,p;\beta,q}^{[1]}(w,w_x)+\cdots,\quad\Omega_{\alpha,p;\beta,q}^{[k]}\in\mathcal{A}_k\]
       satisfy the following equations
\[
 \frac{1}{\ve}(\Lambda-1)\Omega_{\alpha,p;\beta,q}=\frac{\partial h_{\alpha,p}}{\partial t^{\beta,q}},\quad (\al, p), (\beta, q)\in 
\mathcal{I}.\]
Since these functions are symmetric with respect to there indices, i.e.,
\[\Omega_{\al,p;\beta,q}=\Omega_{\beta,q;\al,p},\]
we can define, for any solution $w^1(t), w^2(t)$ of \eqref{w-1}, a tau function $\tau(t)$ such that
\begin{equation}\label{tauf}
\left.\Omega_{\alpha,p;\beta,q}(w,w_x,\dots)\right|_{w^\al=w^\al(t)}=\epsilon^2\frac{\partial \log\tau(t)}{\partial t^{\alpha,p}\partial t^{\beta,q}},\quad (\al, p), (\beta, q)\in 
\mathcal{I}.
\end{equation}
From this definition of the tau function, we have relations
 \[h_{\alpha,p}=\epsilon(\Lambda-1)\frac{\partial \log \tau}{\partial t^{\alpha,p}},\]
 and
     \begin{equation}\label{zh-21}
        w^1=\epsilon(\Lambda-1)\frac{\partial \log \tau}{\partial t^{2,0}},\quad w^2-\log(e^{w^2}-w^1)=(\Lambda-1)(1-\Lambda^{-1})\log \tau.
    \end{equation}
We call 
\[f_{\al,p}=\ve\frac{\p\log\tau}{\p t^{\al,p}}, \quad (\al, p)\in 
\mathcal{I}\]
the one-point functions of the extended Ablowitz-Ladik hierarchy \eqref{w-1}, and the system given by the flows
\begin{equation}\label{zh-33}
\ve \frac{\p f_{\al,p}}{\p t^{\beta,q}}=\Omega_{\al,p;\beta,q},\quad
\frac{\p w^{\gamma}}{\p t^{\beta,q}}=\mathcal{P}_0^{\gamma\xi}\frac{\delta H_{\beta,p}}{\delta w^{\xi}},\quad \gamma=1,2,\, (\al,p), (\beta,q)\in\mathcal{I}
\end{equation}
the tau cover of \eqref{w-1}.

Before finishing this section, let us give the explicit expressions of the flows 
$\frac{\p}{\p t^{1,0}}$, $\frac{\p}{\p t^{2,0}}$, $\frac{\p}{\p t^{0,-1}}$. 
In terms of the unknown functions $P, Q$,
we know from \eqref{lax1}, \eqref{lax2} that the 
flows $\frac{\p}{\p t^{2,0}}$, $\frac{\p}{\p t^{0,-1}}$ have the expressions
\begin{align*}
&\ve\frac{\p P}{\p t^{2,0}}=P(Q^+-Q),\quad \ve\frac{\p Q}{\p t^{2,0}}=Q (Q^+-Q^--P+P^-),\\
&\ve\frac{\p P}{\p t^{0,-1}}=\frac{Q^+}{P^+}-\frac{Q}{P^-},\quad \ve\frac{\p Q}{\p t^{0,-1}}=\frac{Q}{P}-\frac{Q}{P^-},
\end{align*}
where $P^\pm=\Lambda^{\pm 1}P$, $Q^\pm=\Lambda^{\pm 1}Q$. For the flow $\frac{\p}{\p t^{1,0}}$ we have the following theorem.
\begin{theorem}\label{zh-45}
The symmetry of the extended Ablowitz-Ladik hierarchy \eqref{w-1} induced by the $\left(\frac{\p}{\p t^{1,0}}-\frac{\p}{\p x}\right)$-flow is given by the auto-B\"acklund transformation
\begin{align}\label{miura1}
    \tilde{P}&=\frac{P^{-}(Q^{+}-P)}{Q-P^{-}}=P+\epsilon P_1+\epsilon^2P_2+\dots,\\\label{miura2}
    \tilde{Q}&=\frac{Q(Q^{+}-P)}{Q-P^{-}}=Q+\epsilon Q_1+\epsilon^2 Q_2+\dots
\end{align}
of the integrable hierarchy, where $P_i,Q_i$ are differential polynomials of degree $i$. In other words, 
\[e^{\ve \left(\frac{\p}{\p t^{1,0}}-\frac{\p}{\p x}\right)}P=\tilde{P},\quad 
e^{\ve \left(\frac{\p}{\p t^{1,0}}-\frac{\p}{\p x}\right)}Q=\tilde{Q},\]
and for any solution $(P, Q)$ to the extended Ablowitz-Ladik hierarchy, $(\tilde{P}, \tilde{Q})$ is also a solution to this integrable hierarchy.
\end{theorem}
\begin{proof}
Let us denote by 
\[X=A\frac{\p}{\p P}+B \frac{\p}{\p Q}+\left(\p_x A\right)\frac{\p}{\p P_x}+\left(\p_x B\right) \frac{\p}{\p Q_x}+\dots\]
the bihmiltonian vector field given by the flow $\frac{\p}{\p t^{1,0}}-\frac{\p}{\p x}$, i.e.,
\[\frac{\p P}{\p t^{1,0}}-\frac{\p P}{\p x}=A=A^{[0]}+\ve A^{[1]}+\dots,\quad 
\frac{\p Q}{\p t^{1,0}}-\frac{\p Q}{\p x}=B=B^{[0]}+\ve B^{[1]}+\dots,\]
where $A^{[k]}, B^{[k]}$ are differential polynomials of degree $k+1$.
From the expression $\theta_{1,1}$ given in \eqref{zh-20} we know that the 
leading terms of the differential polynomials $A$, $B$ are given by
\[A^{[0]}=\frac{Q_xP-P_xQ}{Q-P},\quad B^{[0]}=\frac{(Q_x-P_x)Q}{Q-P}.\]
 
On the other hand, from \cite{liu2011jacobi} we know the existence of a vector field $Y$ of the form
\[Y=\tilde{A}\frac{\p}{\p P}+\tilde{B} \frac{\p}{\p Q}+\left(\p_x \tilde{A}\right)\frac{\p}{\p P_x}+\left(\p_x \tilde{B}\right) \frac{\p}{\p Q_x}+\dots\]
such that 
\[
    \tilde{P}=e^{\epsilon Y}P,\quad \tilde{Q}=e^{ \epsilon Y}Q.
\]
It is easy to check that the leading terms $\tilde{A}^{[0]}, \tilde{B}^{[0]}$ of the differential polynomials $\tilde{A}, \tilde{B}$ coincide with ${A}^{[0]}, {B}^{[0]}$. One can also verify that the Miura-type transformation given by \eqref{miura1}, \eqref{miura2} keeps the bihamiltonian structure $(\mathcal{P}_0, \mathcal{P}_1)$ invariant. In fact, the Hamiltonian operators $\mathcal{P}_0, \mathcal{P}_1$ can be represented in terms of the unknown functions $P, Q$ as follows \cite{li2022tri}:
	\begin{equation}\label{zh-23}
	    \mathcal{P}_0=\ve^{-1}
	\begin{pmatrix}Q \Lambda ^{-1}-\Lambda Q&(1-\Lambda)Q
\\Q(\Lambda^{-1}-1)& 0\end{pmatrix},	
		\quad\mathcal{P}_1=\ve^{-1}
		\begin{pmatrix}0&P(\Lambda-1)Q\\
		Q(1-\Lambda^{-1})P&Q(\Lambda-\Lambda^{-1})Q
	\end{pmatrix},
	\end{equation}
then we can check that
\[
J\mathcal{P}_0J^{*}=\left.\mathcal{P}_0\right|_{P\to \tilde{P},\,Q\to\tilde{Q}},\quad
J\mathcal{P}_1J^{*}=\left.\mathcal{P}_1\right|_{P\to \tilde{P},\,Q\to\tilde{Q}},
\]
where
\[
  J=  \begin{pmatrix}
        \frac{\partial \tilde{P}}{\partial P}+\frac{\partial\tilde{P}}{\partial P^{-}}\Lambda^{-1}&\frac{\partial \tilde{P}}{\partial Q}+\frac{\partial\tilde{P}}{\partial Q^{+}}\Lambda\\
    \frac{\partial \tilde{Q}}{\partial P}+\frac{\partial\tilde{Q}}{\partial P^{-}}\Lambda^{-1}&\frac{\partial \tilde{Q}}{\partial Q}+\frac{\partial\tilde{Q}}{\partial Q^{+}}\Lambda
    \end{pmatrix},\quad J^{*}=\begin{pmatrix}
        \frac{\partial \tilde{P}}{\partial P}+\Lambda\frac{\partial\tilde{P}}{\partial P^{-}}&
    \frac{\partial \tilde{Q}}{\partial P}+\Lambda\frac{\partial\tilde{Q}}{\partial P^{-}}\\\frac{\partial \tilde{P}}{\partial Q}+\Lambda^{-1}\frac{\partial\tilde{P}}{\partial Q^{+}}&\frac{\partial \tilde{Q}}{\partial Q}+\Lambda^{-1}\frac{\partial\tilde{Q}}{\partial Q^{+}}
    \end{pmatrix}.
    \]
Thus $Y$ is a symmetry of the bihamiltonian structrure 
$(\mathcal{P}_0, \mathcal{P}_1)$. From the triviality of the first Hamiltonian cohomolgy groups \cite{magri, normal, getzler} of the Hamiltonian structures of hydrodynamic type given by the leading terms of $\mathcal{P}_0, \mathcal{P}_1$, it follows that $Y$ is also a bihamiltonian vector field
with respect to the bihamiltonian structrure 
$(\mathcal{P}_0, \mathcal{P}_1)$. Since the leading terms of the bihamiltonian vector fields $X$ and $Y$ coincide, it follows from the uniqueness of bihamiltonian vector fields \cite{dubrovin2006hamiltonian, liu2018lecture, dubrovin2018bihamiltonian} that $X=Y$. The theorem is proved.
\end{proof}

Denote $\tilde{\Lambda}=e^{\ve\frac{\p}{\p t^{1,0}}}$, then we can represent $\tilde{P}, \tilde{Q}$ in the form
\[\tilde{P}=\tilde{\Lambda}\Lambda^{-1} P,\quad \tilde{Q}=\tilde{\Lambda}\Lambda^{-1} Q,\]
and we have
\begin{align}
 & \Lambda^{-1} (\tilde{\Lambda}-1)(\log Q-\log P)=\log \tilde{Q}-\log\tilde{P}-\log Q^{-}+\log P^{-}\\
=& \log\left(\frac{Q(Q^{+}-P)}{Q-P^{-}}\frac{Q-P^{-}}{P^{-}(Q^{+}-P)}\right)-\log Q^{-}+\log P^{-}\\
=&\Lambda^{-1}(\Lambda-1)\log Q.
\end{align}
Thus we arrive at the identity
\[(\tilde{\Lambda}-1)(\log Q-\log P)=(\Lambda-1)\log Q,\]
and by using \eqref{zh-21} we obtain the relation
\begin{equation}\label{4.1}
\log Q=\left(1-\Lambda^{-1}\right)\left(\tilde\Lambda-1\right)\log\tau.
\end{equation}
From this relation we know that Conjecture 11.2 of \cite{liu2022generalized}
follows from the Main Theorem given in Sect.\,\ref{zh-22}.

 \section{The super tau-cover of the extended Ablowitz-Ladik hierarchy}
 
In order to show that the extended Ablowitz-Ladik hierarchy possesses an infinite set of Virasoro symmetries which act linearly on its tau function, we 
construct in this section the super tau-cover of this integrable hierarchy by following the notation and construction given in \cite{liu2021super}. To this end, we first recall some notations associated with the infinite jet space of a super manifold, see  \cite{liu2011jacobi, liu2013bihamiltonian, liu2018lecture} for more details.

Let $M$ be a smooth manifold of dimension $n$, and  $\hat{M}=\prod(T^*M)$ be the super manifold of dimension $(n|n)$ obtained from the cotangent bundle of $M$ by reversing the parity of its fibers. Locally, we consider an open subset $U\subset M$ with coordinates $u^1,\dots,u^n$. The dual coordinates $\theta_1,\dots, \theta_n$ on the fibers of $\hat{U}\subset \hat{M}$ satisfy the relations
\[\theta_i\theta_j+\theta_j\theta_i=0,\quad i, j=1,\dots,n.\]
Let $J^{\infty}(\hat{M})$ be the infinite jet space of $\hat{M}$, and $\hat{\mathcal{A}}$, $\mathcal{A}$ be space of differential polynomials on $\hat{M}$ and $M$ respectively.  Locally $\mathcal{A}, \hat{\mathcal{A}}$ can be represented in the forms 
\begin{align}
\mathcal{A}&=C^{\infty}(\hat{U})[[u^{i,s}\mid i=1,\dots n;s\ge 1]],\\\hat{\mathcal{A}}&=C^{\infty}(\hat{U})[[u^{i,s},\theta^{t}_i\mid i=1,\dots n;s\ge 1,t\ge 0]],\label{zh-26}
\end{align}
where $u^{i,0}=u^{i}$, $\theta_{i}^0=\theta_{i}$. By using the global vector field 
\begin{equation}
    \partial_x=\sum_{i=1}^{n}\sum_{s\ge 0}\left(u^{i,s+1}\frac{\partial}{\partial u^{i,s}}+\theta_{i}^{s+1}\frac{\partial }{\partial\theta_{i}^{s}}\right)
\end{equation}
on $J^{\infty}(\hat{M})$ we define the space of local functionals as the quotient space
\[\hat{\mathcal{F}}=\hat{\mathcal{A}}/\partial_x\hat{\mathcal{A}},\] and denote by $\int$ the projection operator. There are two gradations on $\hat{\mathcal{A}}$ which are defined by 
\[\deg_x u^{i,s}=\deg_x \theta_i^s=s;\quad\deg_{\theta} u^{i,s}=0\,\deg_{\theta} \theta_i^{s}=1.\]
We call them the differential gradation and super gradation of $\hat{\mathcal{A}}$ repsectively, and denote the homogeneous spaces with respect to these gradations by 
\[\hat{\mathcal{A}}_d=\{f\in\hat{\mathcal{A}}|\deg_x f=d\},\quad\hat{\mathcal{A}}^p=\{f\in\hat{\mathcal{A}}|\deg_{\theta} f=p\},\quad \hat{\mathcal{A}}_d^p=\hat{\mathcal{A}}_d\cap\hat{\mathcal{A}}^p.
\] 
From the above definition we have $\hat{\mathcal{A}}^0=\mathcal{A}$, $\hat{\mathcal{A}}^0_d=\mathcal{A}_d$. The
differential and super gradations on $\hat{\mathcal{A}}$ induce two gradations on $\hat{\mathcal{F}}$, and we denote the corresponding homogeneous spaces by $\hat{\mathcal{F}}_d$, $\hat{\mathcal{F}}^p$, $\hat{\mathcal{F}}_d^p$. 
We equip the space of local functionals a graded Lie algebra structure by using the following Schouten-Nijenhuis bracket:
\[[F,G]=\int \left(\frac{\delta F}{\delta\theta_i}\frac{\delta  G}{\delta u^i}+(-1)^p\frac{\delta F}{\delta u^i}\frac{\delta  G}{\delta\theta_i}\right),\quad F\in\hat{\mathcal{F}}^p,\ G\in\hat{\mathcal{F}}^q.\]
Here the variational derivatives are defined by
\begin{equation}
\frac{\delta F}{\delta u^i}=\sum_{s\ge 0}(-\p_x)^{s}\frac{\p\tilde{F}}{\p u^{i,s}},\quad \frac{\delta F}{\delta\theta_i}=\sum_{s\ge 0}(-\p_x)^{s}\frac{\p\tilde{F}}{\p \theta_i^s},
\end{equation}
with $\tilde{F}\in \hat{\mathcal{A}}$ being any lift of $F$, i,e., $F=\int \tilde{F}$.
For each local functional $F\in \hat{\mathcal{F}}^p$ there is associated a derivation $D_F\in \textrm{Der}(\hat{\mathcal{A}})$ which is defined by
\begin{equation}
    D_F=\sum_{i=1}^n\sum_{s\ge 0}\partial_x^s\left(\frac{\delta F}{\delta \theta_i}\right)\frac{\partial}{\partial u^{i,s}}+(-1)^p\partial_x^s\left(\frac{\delta F}{\delta u^i}\right)\frac{\partial}{\partial \theta_i^s}.
\end{equation}
These derivations satisfy the following relations:
\begin{align}
&\left[F,G\right]=\int D_F(\tilde{G}),\quad \forall\,F\in\hat{\mathcal{F}}^p,\ G\in\hat{\mathcal{F}}^q,\label{w-1b}\\
&(-1)^{p-1}D_{[F,G]}=D_F\circ D_G-(-1)^{q-1}(-1)^{p-1}D_G\circ D_F\label{w0}.
\end{align}

A local functional $I\in \hat{\mathcal{F}}^2$ is called a Hamiltonian structure if $[I, I]=0$. We can represent a Hamiltonian structure $I$ in the form
\[I=\frac12\int P^{\al\beta}_s \theta_\al\theta^s_\beta,\]
where $P^{\al\beta}_s\in\mathcal{A}$ are determined by
\[\frac{\delta I}{\delta\theta_\al}=P^{\al\beta}_s\theta_\beta^s.\]
Such a Hamiltonian structure corresponds to a Hamiltonian operator of the form
\[\mathcal{P}=\sum_{s\ge 0}P^{\al\beta}_s\p_x^s\]
and vice versa.
Two Hamiltonian structures $I_0, I_1$ are said to form a bihamiltonian structure if $[I_0, I_1]=0$.
 
Now let us restrict our discussion to a 2-dimensional manifold $M$ with local coordinates $u^1=P, u^2=Q$. We still use $\theta_1, \theta_2$ to denote the dual odd coordinates on $\hat{M}$. We denote by $I_0, I_1$ the local functionals given by the Hamiltonian operators $\mathcal{P}_0, \mathcal{P}_1$ defined in \eqref{zh-23}, i.e., 
\begin{align}
	I_0&=\frac{1}{2\ve}\int\begin{pmatrix}\theta_1&\theta_2\end{pmatrix}\begin{pmatrix}Q \Lambda ^{-1}-\Lambda Q&(1-\Lambda)Q
	\\Q(\Lambda^{-1}-1)& 0\end{pmatrix}\begin{pmatrix}\theta_1\\ \theta_2
	\end{pmatrix},\label{i0}\\[6pt]
	I_1&=\frac{1}{2\ve}\int\begin{pmatrix}\theta_1&\theta_2\end{pmatrix}	\begin{pmatrix}0&P(\Lambda-1)Q\\
	Q(1-\Lambda^{-1})P&Q(\Lambda-\Lambda^{-1})Q
	\end{pmatrix}\begin{pmatrix}\theta_1\\ \theta_2
	\end{pmatrix}.	\label{i1}
	\end{align}
We introduce, following the construction of the super extension of a bihamiltonian integrable hierarchy given in \cite{liu2021super, 2109.01845}, the following family of odd variables:
	\[
	\{\sigma_{\alpha,k}^{s}\mid\alpha=1,2;k\in\mathbb{Z},s\ge0\}.
	\]
We will also denote $\sigma_{\al,k}^0$ by $\sigma_{\al,k}$, and we require that $\sigma_{\alpha,0}=\theta_{\alpha}$.
Let us replace the vector field $\partial_x$ on $J^\infty(\hat{M})$ by
	\begin{equation}\label{dx}
	\partial_x=\sum_{s\ge 0}\left(P^{(s+1)}\frac{\partial}{\partial P^{(s)}}+Q^{(s+1)}\frac{\partial}{\partial Q^{(s)}}\right)+\sum_{i=1}^2\sum_{s\ge0,\,k\in\mathbb{Z}}\sigma_{i,k}^{s+1}\frac{\partial}{\partial \sigma_{i,k}^s},
	\end{equation}
then the odd variables $\sigma_{\al,k}^s$ satisfy the recursion relations
	\begin{equation}\label{super-rr}
	    	\mathcal{P}_0^{\alpha\beta}\sigma_{\beta,k+1}=\mathcal{P}_1^{\alpha\beta}\sigma_{\beta,k},
	\end{equation}
which can be represented in the following explicit form:	
	\begin{align}\label{recursions}
&(\Lambda Q-Q\Lambda^{-1})\sigma_{1,k+1}+(\Lambda-1)Q\sigma_{2,k+1}-P(1-\Lambda)Q\sigma_{2,k}=0,\\
&\sigma_{1,k+1}+P\sigma_{1,k}+(\Lambda+1)Q\sigma_{2,k}=0,\label{recursions2}
	\end{align}
	where $k\in\mathbb{Z}$. 
We also extend the ring of differential polynomials $\hat{\mathcal{A}}$ to the ring 
\begin{equation}\label{zh-27}
\mathcal{B}=C^{\infty}(U)[[P^{(s)},\,Q^{(s)},\, \sigma_{\alpha,k}^s\mid \alpha=1,2;\, k\in\mathbb{Z},\, s\ge 0]]/J,
\end{equation}
where $J$ is the differential ideal generated by the left-hand sides of recursion relations \eqref{recursions}, \eqref{recursions2}. 
	 
The derivations $D_{I_0}, D_{I_1}\in \textrm{Der}(\hat{\mathcal{A}})$ given by the bihamiltonian structure $(I_0, I_1)$ yield two odd flows
\begin{equation}\label{zh-24}
\frac{\p u^{\al}}{\p \tau_i}=\frac{\delta I_i}{\delta \theta_{\al}}=D_{I_i}(u^{\al}),\quad \frac{\p \theta_{\al}}{\p \tau_i}=\frac{\delta I_i}{\delta u^{\al}}=D_{I_i}(\theta_{\al}),\quad i=0,1.
\end{equation}
They commute with each other, and satisfy the recursion relation
\begin{equation}\label{zh-25}
	\frac{\p u^{\al}}{\p\tau_1}=\left(\mathcal{P}_1\circ\mathcal{P}_0^{-1}\right)^{\al}_{\gamma}\frac{\p u^{\gamma}}{\p \tau_0}=R^{\al}_{\gamma}\frac{\p u^{\gamma}}{\p \tau_0},
\end{equation}
where the recursion operator $\mathcal{R}=(R^{\al\gamma})=\mathcal{P}_1\circ\mathcal{P}_0^{-1}$ has the expression
	\[
	\mathcal{R}=\begin{pmatrix}
	-P&P(Q-\Lambda Q\Lambda)(1-\Lambda)^{-1}Q^{-1}\\
	-Q(1+\Lambda^{-1})& Q((\Lambda-1)P^{-}+(1+\Lambda^{-1})(Q-\Lambda Q\Lambda))(1-\Lambda)^{-1}Q^{-1}
	\end{pmatrix}.
	\]
We may continue the recursion procedure  \eqref{zh-25} to define odd flows $\frac{\p}{\p \tau_2}$,
$\frac{\p}{\p \tau_3}$ and so on, these flows should correspond to the Hamiltonian structures $\mathcal{R}\mathcal{P}_1$, $\mathcal{R}^2\mathcal{P}_1, \dots$.
However, due to the appearance of the integral operator $(1-\Lambda)^{-1}$ in the expression of the recursion operator $\mathcal{R}$, these odd flows are in general nonlocal, i.e., they can not be represented in terms of elements of $\hat{\mathcal{A}}$ that is defined in \eqref{zh-26}. It is shown in 
\cite{liu2011jacobi, liu2013bihamiltonian, liu2018lecture} 
that such flows are important for our understanding of properties of the Virasoro symmetries of the integrable hierarchy associated with the bihamiltonian structure $(I_0, I_1)$. The introduction of the odd variables $\sigma_{\al,k}^s$ enables us to define such flows in a consistent way, and
we obtain a certain super extension of the extended Ablowitz-Ladik hierarchy, which is presented in the following proposition.

	\begin{pro}\label{sp-extension}
		The super extension of the extended Ablowitz-Ladik hierarchy \eqref{w-1} consists of the following flows: 
		\begin{enumerate}
		\item The flows of extended Ablowitz-Ladik hierarchy \eqref{w-1} represented in terms of the unknown functions $P, Q$ and the two point functions $\Omega_{\al,p;\beta,q}$:
	\begin{align}
	\ve \frac{\partial P}{\partial t^{\alpha,k}}&=\frac{P}{Q-P}\left((\Lambda-1)\Omega_{2,0;\al,k}-Q(1-\Lambda^{-1})h_{\al,k}\right),\\
	\ve \frac{\partial Q}{\partial t^{\alpha,k}}&=\frac{Q}{Q-P}\left((\Lambda-1)\Omega_{2,0;\al,k}-P(1-\Lambda^{-1})h_{\al,k}\right).
	\end{align}
 \item The odd flows correspond to the local and nonlocal Hamiltonian structures:
 \begin{align}
     \label{odd-flow1}
	   & \ve\frac{\partial P}{\partial \tau_k}=P(\Lambda-1)Q\sigma_{2,k-1},
			\quad
		\ve	\frac{\partial Q}{\partial \tau_k}=Q(\Lambda^{-1}-1)\sigma_{1,k},\\
			\label{odd-flow2}
			&\ve\frac{\partial \sigma_{1,k+m}}{\partial \tau_k}=\sum\limits_{i=0}^{m-1}\sigma_{1,k+i}(1-\Lambda)Q\sigma_{2,k+m-1-i},\quad m\ge 1,\\
			&\frac{\partial \sigma_{1,k}}{\partial \tau_{k}}=0,\quad
			\frac{\partial \sigma_{1,k}}{\partial \tau_{k+m}}=-\frac{\partial \sigma_{1,k+m}}{\partial \tau_k},\quad m\ge 1,\\
			& \ve\frac{\partial \left(Q\sigma_{2,k+m}\right)}{\partial \tau_k}=-Q\sum\limits_{i=0}^{m}\sigma_{1,k+m-i}^-\sigma_{1,k+i},\quad m\ge 0,\\\label{odd-flow3}
			& \frac{\partial \left(Q\sigma_{2,k}\right)}{\partial \tau_{k+1}}=0,\quad \ve\frac{\partial \left(Q\sigma_{2,k}\right)}{\partial \tau_{k+m}}=Q\sum\limits_{i=1}^{m-1}\sigma_{1,k+m-i}^-\sigma_{1,k+i},\quad m\ge 2.
 \end{align}
 Here $k\in\mathbb{Z}$, and $\sigma_{\al,k}^{-}=\Lambda^{-1}\sigma_{\al,k}$.
 \item  The evolutions of the odd unknown functions along the time variables $t^{\al,p}$ are defined by the following equations:
		\begin{equation}\label{odd-flow4}
		    \frac{\partial Q(\sigma_{1,k}+\sigma_{2,k})}{\partial t^{\beta,q}}=\frac{\partial \Omega_{2,0;\beta,q}}{\partial\tau_k},\quad
		     \frac{\partial (P\sigma_{1,k}+Q\sigma_{2,k})}{\partial t^{\beta,q}}=\frac{\partial h_{\beta,q} }{\partial \tau_{k+1}}.
		\end{equation}
	\end{enumerate}
		All these flows are well defined and commute with each other.
	\end{pro}
\begin{proof}
By using the recursion relations \eqref{recursions}, \eqref{recursions2}, \eqref{zh-25}  and by induction on $k$ we can get the expressions of the flows \eqref{odd-flow1}--\eqref{odd-flow3} starting from \eqref{zh-24}. 
 The equations \eqref{odd-flow4}
can be obtained from the commutation relations
\[\frac{\partial^2 P}{\partial \tau_k\partial t^{\alpha,p}}=\frac{\partial^2 P}{\partial t^{\alpha,p}\p \tau_k},\quad\frac{\partial^2 Q}{\partial \tau_k\partial t^{\alpha,p}}=\frac{\partial^2 Q}{\partial t^{\alpha,p}\p \tau_k}.\] 
It's easy to check that all the flows given above are well defined and mutually commutative.
\end{proof}

 We are to use the shift operator $T$ which is introduced in \cite{liu2023variational} to describe the relation among the odd flows of the super extension of a bihamiltonian integrable hierarchy. It acts on elements of $\hat{\mathcal{A}}$ and  only shifts the index of odd variables, e.g.,
\[T(a\sigma_{\alpha,k})=a\sigma_{\alpha,k+1},\quad T^{-1}(a\sigma_{\alpha,k})=a\sigma_{\alpha,k-1},\quad a\in\mathcal{A}.\]
From \eqref{odd-flow1} it follows that 
\begin{equation}\label{w10}
\frac{\p b}{\p \tau_k}=T^{k}\frac{\p b}{\p \tau_0},\quad k\in\mathbb{Z}
\end{equation}
for any differential polynomial $b\in\mathcal{A}$.
By using the definition \eqref{i0}, \eqref{i1} of $I_0$, $I_1$ we can check that 
\begin{equation}\label{w9}
	T\frac{\delta I_0}{\delta \theta_i}=\frac{\delta I_1}{\delta \theta_i},\quad i=1,2, 
\end{equation}
therefore for any $(\al,k)\in\mathcal{I}$ we have
\begin{equation}\label{w8}
[I_1,H_{\alpha,k}]=T[I_0, H_{\alpha,k}].
\end{equation}

Let us proceed to construct the super tau-cover of the extended Ablowitz-Ladik hierarchy. 
By using a similar method  employed in the proof of  Lemma 5 given in \cite{2109.01845}, we know the existence of differential polynomials $\Delta_{\alpha,p}^{n,m}\in\mathcal{B}$ such that 
\[\epsilon\frac{\partial^2h_{\alpha,k}}{\partial \tau_m\partial\tau_n}=(\Lambda-1)\Delta_{\alpha,k}^{n,m},\quad \Delta_{\alpha,k}^{n,m}=-\Delta_{\alpha,k}^{m,n},\quad (\al,k)\in\mathcal{I},\,m,n\in\mathbb{Z},\]
here $\mathcal{B}$ is defined in \eqref{zh-27}. 
We introduce two families of the odd variables $\psi_{0,1}^n$, $\psi_{1,1}^n$, $n\in\mathbb{Z}$ which satisfy the relations
  \[(\Lambda-1)\psi_{0,1}^n=\epsilon\frac{\partial h_{0,1}}{\partial \tau_n},\quad (\Lambda-1)\psi_{1,1}^n=\epsilon\frac{\partial h_{1,1}}{\partial \tau_n},\] 
and define the evolutions of these odd variables with respect to $t^{\beta,q}, \tau_n$ by     \begin{equation}
      \frac{\partial\psi_{\alpha,1}^n}{\partial t^{\beta,q}}=\frac{\partial \Omega_{\alpha,1;\beta,q}}{\partial \tau_n}, \quad
      \frac{\partial \psi_{\alpha,1}^n}{\partial \tau_m}=\Delta_{\alpha,1}^{n,m},\quad (\beta,q)\in\mathcal{I},\,\al=0,1,\,m,n\in\mathbb{Z}.
  \end{equation}
\begin{pro} For any given $(\alpha,k)\in\mathcal{I}$,  $(\alpha,k)\neq (0,1),(1,1)$, there exist differential polynomials $\psi_
{\alpha,k}^n\in\mathcal{B}$ such that 
    \[\epsilon\frac{\partial h_{\alpha,k}}{\partial\tau_n}=(\Lambda-1)\psi_{\alpha,k}^n,\quad n\in\mathbb{Z}.\]
\end{pro}
\begin{proof}
Since the bihamiltonian conserved quantities $H_{\al,-1}\,(\al=1,2)$, $H_{0,-1}$ are deformations of  Casimirs of the dispersionless limits of  the bihamiltonian structure $(I_0, I_1)$, we know that
\begin{equation}\label{zh-28}
    [I_0, H_{\alpha,-1}]=0,
    \quad [I_1,H_{0,-1}]=0,
    \quad\al=1,2.
\end{equation}  
By using \eqref{w-1b} we have
\begin{align*}
\int \frac{\p h_{\al,k}}{\p \tau_0}&=\int D_{I_0}(h_{\al,k})=	[I_0,H_{\al,k-1}],
\\
\int \frac{\p h_{\al,k}}{\p \tau_1}&=\int D_{I_1}(h_{\al,k})=	[I_1,H_{\al,k-1}],\quad (\al,k)\in\mathcal{I}.
\end{align*}
Thus from \eqref{w10}, \eqref{w8} it follows that 
\begin{align}
    \int \frac{\partial h_{\alpha,0}}{\partial \tau_n}&=T^n[I_0,H_{\alpha,-1}]=0,\quad\al=1,2,\\
    \int \frac{\p h_{0,0}}{\p \tau_n}&=T^{n-1} [I_1,H_{0,-1}]=0,\quad n\in\mathbb{Z}.
\end{align}
Since a bihamitonian vector field is uniquely determined by its leading term, we obtain from \eqref{zh-29}--\eqref{zh-30} the following bihamiltonian recursion relations: 
    \begin{align}
     & [I_1,H_{2,k-1}]=(k+1)[I_0,H_{2,k}],\quad [I_1,H_{1,k-1}]=k[I_0,H_{1,k}]+2[I_0,H_{2,k-1}],\label{w6}\\
  & [I_1,H_{0,k-1}]=k[I_0,H_{0,k}]+[I_0,H_{2,k-1}],
    \quad [I_1,H_{0,-k-1}]=-k[I_0,H_{0,-k}],\label{w7}
\end{align}
where $k\ge 0$. Then we can derive from these relations and \eqref{w8} that
  \begin{equation}
      \int \frac{\partial h_{\alpha,k}}{\partial \tau_n}=T^{n}[I_0,H_{\al,k-1}]=0,
  \end{equation}
  here $T^{n}[I_0,H_{\al,k-1}]$ is viewed as an element of  $\mathcal{B}/\p_x\mathcal{B}$ with $\p_x$, $\mathcal{B}$ being defined by  \eqref{dx}, \eqref{zh-27} respectively. For example, when $(\al,k)=(2,1)$ we have 
  \begin{align*}
  	\int\frac{\p h_{2,1}}{\p \tau_n}
  	&=T^n\int\frac{\p h_{2,1}}{\p \tau_0}=T^n\int D_{I_0}(h_{2,1})= T^n[I_0,H_{2,0}]\\
  	&=T^n[I_1,H_{2,-1}]=T^{n+1}[I_0,H_{2,-1}]=0.
  \end{align*}
 Thus $\frac{\partial h_{\alpha,k}}{\partial \tau_n}$ is a total $x$-derivative for any $(\al,k)\in\mathcal{I}$ and $(\Lambda-1)^{-1}\frac{\partial h_{\alpha,k}}{\partial \tau_n}$ is a differential polynomial of $\mathcal{B}$.
The proposition is proved. \end{proof}

Now we are ready to give a super extension of the tau cover 
\eqref{zh-33} of the extended Ablowitz-Ladik hierarchy.
\begin{de}
We call the system consists of the flows 
	\begin{align}\label{1pt}
\ve\frac{\partial f_{\alpha,p}}{\partial t^{\beta,q}}&=\Omega_{\alpha,p;\beta,q},\quad
\ve \frac{\partial f_{\alpha,p}}{\partial \tau_n}=\psi_{\alpha,p}^n,\\
\frac{\partial \psi_{\alpha,p}^n}{\partial t^{\beta ,q}}&=\frac{\partial \Omega_{\alpha,p;\beta,q}}{\partial \tau_n},\quad
\frac{\partial \psi_{\alpha,p}^n}{\partial \tau_m}=\Delta_{\alpha,p}^{n,m},
\quad (\al,p)\in\mathcal{I}, n\in\mathbb{Z},
\end{align}
and the ones given in Proposition \ref{sp-extension} the super tau-cover of the extended Ablowtiz-Ladik hierarchy.
\end{de}

The following proposition shows that the odd flows \eqref{odd-flow1} can be represented in term of Lax equations, which will be used in the next section.
\begin{pro}\label{odd-lax}
	There exist operators $A$, $B$ of the form
\[A=\sum_{i\ge 1} a_i \Lambda^{-i},\quad  B=\sum_{i\ge 1} b_i \Lambda^{-i}\] with $a_i,b_i\in \mathcal{B}$ 
such that
	\begin{equation}\label{zh-31}
	\ve\frac{\partial P}{\partial \tau_0}=(\Lambda-P)B-A(\Lambda-P),\quad\ve\frac{\partial Q}{\partial \tau_0}=(\Lambda-Q)B^{-}-A(\Lambda-Q),
	\end{equation}
here $B^{-}=\Lambda^{-1}B\Lambda$. Moreover, we have
	\begin{equation}\label{zh-32}
	\ve\frac{\partial L}{\partial \tau_n}=[T^nB,L].
	\end{equation}
\end{pro}
\begin{proof}
In order to show that we can represent the flow $\frac{\p}{\p\tau_0}$ in the form \eqref{zh-31}, we need to find differential polynomials $a_i$, $b_i\in\mathcal{B}$ such that 
	\begin{align*}
	b_1^+-a_1&=-Q^+\sigma_{1,0}^++Q\sigma_{1,0}^--Q^+\sigma_{2,0}^++Q\sigma_{2,0},\quad
	b_1-a_1=-Q\sigma_{1,0}+Q\sigma_{1,0}^-,\\
	 b_k^+-a_k&=P b_{k-1}-P^{-(k-1)}a_{k-1},\quad b_k-a_k=Qb_{k-1}^--Q^{-(k-1)}a_{k-1},\quad k\ge 2.
	\end{align*}
Indeed, the following choice of $a_i,b_i$ meets these conditions:
	 \begin{align*}
	a_1&=-Q(\sigma_{1,0}^-+\sigma_{2,0}),\quad b_1=-Q(\sigma_{1,0}+\sigma_{2,0}),\\
	  a_2&=Q\sigma_{1,1}+Q\sigma_{2,1}+QP^-\sigma_{2,0}-QQ^-\sigma_{2,0}^-,\\
	 b_2&=Q\sigma_{1,1}+Q\sigma_{2,1}+QP^-\sigma_{2,0}-QQ^-\sigma_{2,0},\\
	 a_k&=a_{k-1} \Lambda^{-k+1}P+Ta_{k-1}-Ta_{k-2} \Lambda^{-k+2}Q,\\
	 b_k&=P^-b_{k-1}^-+Tb_{k-1}^--Q^-Tb_{k-2}^{--},\quad k\ge 3.
	 \end{align*}
Since $L=\Lambda(\Lambda-Q)^{-1}(\Lambda-P)$,  we have
	\begin{align*}
	\ve\frac{\partial L}{\partial \tau_0}=&\Lambda(\Lambda-Q)^{-1}\ve\frac{\partial Q}{\partial \tau_0}(\Lambda-Q)^{-1}(\Lambda-P)-\Lambda(\Lambda-Q)^{-1}\ve\frac{\partial P}{\partial \tau_0}\\
	=&\Lambda(\Lambda-Q)^{-1}((\Lambda-Q)B^--A(\Lambda-Q))(\Lambda-Q)^{-1}(\Lambda-P)\\
	&-\Lambda(\Lambda-Q)^{-1}((\Lambda-P)B-A(\Lambda-P))=[B,L],
	\end{align*}
from which it follows \eqref{zh-32}.
The proposition is proved.	
\end{proof}


\section{Virasoro symmetries of the extended Ablowitz-Ladik hierarchy}
In this section, we show that the extended Ablowitz-Ladik hierarchy
possesses an infinite number Virasoro symmetries which act linearly on its tau function, i.e.,
\begin{equation}\label{zh-36}
	\frac{\p \tau}{\p s_m}=L_m\tau,\quad m\ge -1,
\end{equation}
where the Virasoro operators $L_m$ are given in \eqref{virop}--\eqref{vir-m}.
The flows \eqref{zh-36} induce the following flows of the tau cover \eqref{zh-33} of the  extended Ablowitz-Ladik hierarchy:
\begin{align}
\frac{\p f_{\gamma,k}}{\p s_m}=&a_m^{\al,p;\beta,q}\left(\ve \frac{\p\Omega_{\al,p;\beta,q}}{\p t^{\gamma,k}}
+f_{\al,p}\Omega_{\gamma,k;\beta,q}+f_{\beta,q}\Omega_{\gamma,k;\al,p}\right)\notag\\
&+b^{\al,p}_{m;\gamma,k} f_{\al,p}+\frac1{\ve}b^{\al,p}_{m;\beta,q}t^{\beta,q} \Omega_{\gamma,k;\al,p}+\frac1{\ve}c_{m;\al,p;\gamma,k} t^{\al,p},\label{zh-40}\\
\frac{\p w^1}{\p s_m}=&(\Lambda-1)\left(\frac{\p f_{2,0}}{\p s_m}\right),\quad \frac{\p \bigl(w^2-\log(e^{w^2}-w^1)\bigr)}{\p s_m}=(\Lambda-1)B_-\left(\frac{\p f_{0,0}}{\p s_m}\right),\label{zh-41}
\end{align}
 where $w^1=Q-P, w^2=\log Q$, as they are defined in \eqref{zh-37}, and the operator $B_-$ is given by
\[B_-=\frac{(1-\Lambda^{-1})}{\ve\p_x}=1-\frac{1}{2}\ve\p_x+\dots.\]
We are to prove that these flows are symmetries of the tau cover of the extended Ablowitz-Ladik hierarchy, i.e., the following commutation relations hold true:
\begin{equation}\label{31-w1}
\left[\frac{\p}{\p s_m},\,\frac{\p }{\p t^{\al,p}}\right]X=0,\quad m\ge -1,\ (\al,p)\in\mathcal{I},
\end{equation}
here $X=f_{\beta,q},\,P,\,Q$.
To this end, we employ the approach of \cite{2109.01845} by considering a certain modification of the above defined flows which act on the unknown functions of the super tau cover of the extended Ablowitz-Ladik hierarchy.
The first step is to introduce the modified Virasoro operators
	\begin{equation}
	\tilde{L}_{m}=L_m+L_m^{\text{odd}},\quad L_m^{\text{odd}}=\sum_{p\in\mathbb{Z}}(c_0+p)\tau_{p}\frac{\partial}{\partial\tau_{p+m}},
	\end{equation}	
	where $c_0$ is an arbitrarily given constant.
	It's easy to  check that these modified operators still satisfy the Virasoro commutation relations
	\[
		[\tilde{L}_m,\tilde{L}_n]=(m-n)\tilde{L}_{m+n},\quad m, n\ge -1,
	\]
and they yield the flows
\begin{align}
&\frac{\p f_{\gamma,k}}{\p \tilde{s}_m}=a_m^{\al,p;\beta,q}\left(\ve \frac{\p\Omega_{\al,p;\beta,q}}{\p t^{\gamma,k}}
+f_{\al,p}\Omega_{\gamma,k;\beta,q}+f_{\beta,q}\Omega_{\gamma,k;\al,p}\right)\notag\\
&\hskip 1.2cm   +b^{\al,p}_{m;\gamma,k} f_{\al,p}+\frac1{\ve}b^{\al,p}_{m;\beta,q}t^{\beta,q} \Omega_{\gamma,k;\al,p}+\frac1{\ve}c_{m;\al,p;\gamma,k} t^{\al,p}+\sum_{p\in\mathbb{Z}}(c_0+p)\tau_{p}\frac{\partial f_{\gamma,k}}{\partial\tau_{p+m}},\\
&\frac{\p w^1}{\p \tilde{s}_m}=(\Lambda-1)\left(\frac{\p f_{2,0}}{\p \tilde{s}_m}\right),\quad \frac{\p \bigl(w^2-\log(e^{w^2}-w^1)\bigr)}{\p \tilde{s}_m}=(\Lambda-1)B_-\left(\frac{\p f_{0,0}}{\p \tilde{s}_m}\right),\\
&\frac{\p \bigl(Q\sigma_{1,n}+Q\sigma_{2,n}\bigr)}{\p\tilde{s}_m}=\frac{\p \psi_{2,0}^n}{\p\tilde{s}_m},\quad  \frac{\p\bigl(P\sigma_{1,n-1}+Q\sigma_{2,n-1}\bigr)}{\p\tilde{s}_m}=\Lambda B_-\left(\frac{\p \psi_{0,0}^n}{\p \tilde{s}_m}\right),\\
&\frac{\p\psi_{\al,p}^n}{\p\tilde{s}_m}=\ve\frac{\p}{\p\tau_n}\left(\frac{\p f_{\al,p}}{\p \tilde{s}_m}\right).\label{zh-39}
\end{align}	
	
We can represent the operators $\tilde{L}_{-1}$, $\tilde{L}_2$  in the form
	\begin{align*}
	\tilde{L}_{-1}=&E_{-1}+\frac{1}{\epsilon^2}t^{1,0}t^{2,0},\\
	\tilde{L}_2=&E_2+\epsilon^2\frac{\partial^2}{\partial t^{2,0}\partial t^{2,0}}+2\left(2t^{1,0}+t^{0,0}\right)\frac{\p}{\p t^{2,1}}+\frac{1}{\epsilon^2}\sum_{k\ge 0}\left(3k^2+6k+2\right)(-1)^{k+2}t^{0,-k-2}t^{1,k},
	\end{align*}
where the linear operators $E_{-1}$, $E_2$ are defined by
	\begin{align*}
	E_{-1}=&\sum\limits_{k\ge 1}\left(t^{1,k}\frac{\partial}{\partial t^{1,k-1}}+t^{2,k}\frac{\partial}{\partial t^{2,k-1}}\right)+\sum_{p\in\mathbb{Z}}t^{0,p}\frac{\partial}{\partial t^{0,p-1}}+(p+c_0)\tau_p\frac{\partial}{\partial \tau_{p-1}}\\
	E_2=&\sum_{k\ge 1}\frac{(2+k)!}{(k-1)!}\left(t^{0,k}\frac{\partial}{\partial t^{0,k+2}}-t^{0,-k-2}\frac{\partial}{\partial t^{0,-k}}+t^{1,k}\frac{\partial}{\partial t^{1,k+2}}+t^{2,k-1}\frac{\partial}{\partial t^{2,k+1}}\right)\notag\\
	&+\sum_{k\ge 1}\left(3k^2+6k+2\right)\left(2t^{1,k}+t^{0,k}\right)\frac{\partial}{\partial t^{2,k+1}}-t^{0,-1}\frac{\partial}{\partial t^{2,0}}+\sum_{k\ge 0}(k+c_0)\tau_k\frac{\partial}{\partial \tau_{k+2}},
	\end{align*}
From the above definition of the flows $\frac{\p}{\p\tilde{s}_m}$, we have
the following:
	\begin{enumerate}
		\item For $m=-1$,
		\begin{align}
		\label{first}
		\frac{\partial P}{\partial \tilde{s}_{-1}}&=E_{-1}P-1+\frac{Q^+}{P^+}+\frac{Q}{P^-},\quad \frac{\partial Q}{\partial \tilde{s}_{-1}}=E_{-1}Q+\frac{Q}{P^-},\\
		\frac{\partial\sigma_{1,0}}{\partial \tilde{s}_{-1}}&=E_{-1}\sigma_{1,0}+c_0\sigma_{1,-1},\quad
		\frac{\partial (Q\sigma_{2,0})}{\partial \tilde{s}_{-1}}=E_{-1}Q\sigma_{2,n}+c_0Q\sigma_{2,-1}-\frac{Q}{P^-}\sigma_{1,0}.\label{first2}
		\end{align}
		\item For $m=2$,
		\begin{align}
			\frac{\partial P}{\p \tilde{s}_2}&=E_2P+\frac{P}{P-Q}\left(Q\left(1-\Lambda^{-1}\right)(\Lambda-1)X-(\Lambda-1)Y\right),\label{zh-38}\\
		\frac{\partial Q}{\p \tilde{s}_2}&=E_2Q+\frac{Q}{P-Q}\left(P\left(1-\Lambda^{-1}\right)(\Lambda-1)X-(\Lambda-1)Y\right),\\
		\frac{\partial\sigma_{1,0}}{\p \tilde{s}_2}&=E_2\sigma_{1,0}+c_0\sigma_{1,2}+A_{1,0},\quad
		\frac{\partial (Q\sigma_{2,0})}{\p \tilde{s}_2}=E_2(Q\sigma_{2,0})+c_0Q\sigma_{2,2}+A_{2,0},\label{last} 
		\end{align}
		where
		\begin{align*}
		 X&=2\ve^{-1}\left(t^{0,0}+2t^{1,0}\right)f_{2,1}+\Omega_{2,0;2,0}+f_{2,0}^2,\\
		Y&=2\ve^{-1}\left(t^{0,0}+2t^{1,0}\right)\Omega_{2,1;2,0}+\frac{\partial\Omega_{2,0;2,0}}{\partial t^{2,0}}+2\ve^{-1}f_{2,0}\Omega_{2,0;2,0}+6\ve^{-1}f_{2,2},\\
		A_{1,0}&=\frac{1}{Q-P}\left(\frac{\partial Y}{\partial \tau_0}-\frac{\partial (\Lambda-1)X}{\partial \tau_{1}}-\left((\Lambda-1)Y\right)\sigma_{1,0}-P\sigma_{1,2}-Q\sigma_{2,2}\right),\\
		A_{2,0}&=\frac{1}{P-Q}\left(P\frac{\partial Y}{\partial \tau_0}-Q\frac{\partial(\Lambda-1) X}{\partial \tau_{1}}-PQ\left(
		\left(\left(\Lambda-\Lambda^{-1}\right)X\right)\sigma_{1,0}+\sigma_{1,2}+\sigma_{2,2}\right)\right).
		\end{align*}
	\end{enumerate}	

We prove the commutation relations \eqref{31-w1} by the following three steps:
	\begin{enumerate}
		\item We prove that \begin{equation}
		\left[\frac{\p}{\p \tilde{s}_{i}},\frac{\p}{\p\tau_j}\right]K=0,\quad i=-1,\,2,\, j=0,\,1,
		\end{equation}
			for $K=f_{\al,p}, \sigma_{1,0}, \sigma_{2,0}, P, Q$.
		\item Then we check the validity of the locality condition 
		\begin{equation}
		\left[\frac{\p}{\p \tilde{s}_{i}},\frac{\p}{\p t^{\al,p}}\right]K\in\hat{\mathcal{A}},\quad\text{for 
		}\,i=-1,\,2,\,(\al,p)\in\mathcal{I}
		\end{equation}
	that is introduced in \cite{2109.01845}, where $K=\sigma_{1,0},\,\sigma_{2,0}, P, Q$.
		\item By using the Jacobi identity, the super tau cover of the extended Ablowitz-Ladik hierarchy and by using the main result \cite{2109.01845} we can obtain that
		\begin{equation}
				\left[\frac{\p}{\p s_{i}},\frac{\p}{\p t^{\al,p}}\right]X=0,\quad \text{for }\,i=-1,2,
		\end{equation}
		where $X=f_{\beta,q}, P,Q$.
		\end{enumerate}
These steps are accomplished by the three lemmas that are given below.
\begin{lem}\label{5.1}
	The flows $\frac{\p}{\p \tilde{s}_{-1}}, \frac{\p}{\p \tilde{s}_2}$ which are given by  \eqref{first}--\eqref{last}  commute with the odd flows $\frac{\p }{\p \tau_0}, \frac{\p }{\p \tau_1}$  which are  defined by \eqref{odd-flow1}--\eqref{odd-flow3}, i.e.,
		\begin{equation}
	\left[\frac{\p}{\p \tilde{s}_{i}},\frac{\p}{\p\tau_j}\right]K=0,\quad
	i=-1, 2,\, j=0, 1,\,\label{w5}
\end{equation}
for $K=f_{\al,p}, \sigma_{1,0}, \sigma_{2,0}, P, Q$.
\end{lem}
\begin{proof}
The validity of 
\eqref{w5} for $i=-1$ can be checked via a straightforward calculation, so we only need give the prove of 
\eqref{w5} for $i=2$.

By using the result given in Proposition \ref{odd-lax}, we can obtain the following identities:
\begin{align}
2\ve\frac{\partial f_{2,1}^+}{\partial\tau_1}=&\left(\Lambda-1\right)^{-1}\Lambda\frac{\partial \Res L^2}{\partial \tau_1}\notag\\
=&-\left(Tb_2^++Tb_2^{++}+b_1^+\left(Q^+-P^++Q-P\right)\right),\label{4.2.3}\\
6\ve\frac{\partial f_{2,2}}{\partial\tau_0}=&(\Lambda-1)^{-1}\frac{\partial \Res L^3}{\partial \tau_0}\notag\\
=&-\left(b_3+b_3^++b_3^{++}+b_2\ell_2^{--}+b_2^+\ell_2^-+b_1\ell_1^-\right),\\
2\ve\frac{\partial f_{2,1}^+}{\partial t^{2,0}}=&Q^+(Q-P)^2+Q^+\left(Q^+-P^+\right)(Q-P)\notag\\
&+Q^{+}Q\left(Q^--P^-\right)+Q^{++}Q^{+}(Q-P),\label{4.2.4}
\end{align}
where $\ell_i$ are the coefficients of $\Lambda^i$ in the expansion of $L^3$, i.e., 
\[L^3=\Lambda^3+\ell_2\Lambda^2+\ell_1\Lambda+\ell_0+\dots,\]  
the operators $T$ and $B=b_1\Lambda^{-1}+b_2\Lambda^{-2}+\dots$ are defined in Proposition \ref{odd-lax}. From \eqref{4.2.3}--\eqref{4.2.4} we also have the identities
\begin{align}
(1-\Lambda^{-1})A_{1,0}&=\ve\frac{\partial}{\partial \tau_{0}}\left(\frac{P(1-\Lambda^{-1})(\Lambda-1)X-(\Lambda-1)Y}{P-Q}\right),\\
(\Lambda-1)A_{2,0}&=\ve\frac{\partial}{\partial \tau_{1}}\left(\frac{Q(1-\Lambda^{-1})(\Lambda-1)X-(\Lambda-1)Y}{P-Q}\right).
\end{align}
Then by using these identities and \eqref{zh-38}--\eqref{last} we obtain the commutation relations 
\begin{equation}
\frac{\partial^2\log P}{\p \tilde{s}_2\partial \tau_n}=\frac{\partial^2\log P}{\partial \tau_n\p \tilde{s}_2},\quad \frac{\partial^2\log Q}{\p \tilde{s}_2\partial \tau_n}=\frac{\partial^2\log Q}{\partial \tau_n\p \tilde{s}_2},\end{equation}
for the unknown functions $P, Q$, 
and the following ones for the odd unknown functions: 
\begin{align}\label{spcom}
\frac{\partial^2\sigma_{1,0}}{\partial \tau_0\p \tilde{s}_2}&
=\frac{\partial^2\sigma_{1,0}}{\p \tilde{s}_2\partial \tau_0},\quad \frac{\partial^2(Q\sigma_{1,0}+Q\sigma_{2,0})}{\partial\tau_0\p \tilde{s}_2}
=\frac{\partial^2(Q\sigma_{1,0}+Q\sigma_{2,0})}{\p \tilde{s}_2\partial \tau_0},
\\
\frac{\partial^2 (Q\sigma_{2,0})}{\partial \tau_1\p \tilde{s}_2}&=
\frac{\partial^2(Q\sigma_{2,0})}{\p \tilde{s}_2\partial\tau_1}
,\quad\frac{\partial^2( P\sigma_{1,0}+Q\sigma_{2,0})}{\partial \tau_1\p \tilde{s}_2}=
\frac{\partial^2(P\sigma_{1,0}+Q\sigma_{2,0})}{\p \tilde{s}_2\partial\tau_1}.
\end{align}
From \eqref{zh-39} we also have 
\begin{equation}
\frac{\p}{\p \tilde{s}_j}\left(\frac{\p f_{\al,p}}{\p \tau_i}\right)=\frac{\p\psi_{\al,p}^i}{\p \tilde{s}_j}=\frac{\p}{\p \tau_i}\left(\frac{\p f_{\al,p}}{\p \tilde{s}_j}\right).
\end{equation}
The Lemma is proved.
\end{proof}

For the convenience of our proof of the validity of the locality condition, we use the coordinates $(w^1,w^2,\tilde{\sigma}_{1,n},\tilde{\sigma}_{2,n})$ instead of $(P,Q,\sigma_{1,n},\sigma_{2,n})$, here $w^1, w^2$ are given by \eqref{zh-37} and
	\begin{align*}
	\tilde{\sigma}_{1,n}&=\frac{\p P}{\p w^1}\sigma_{1,n}+\frac{\p Q}{\p w^1}\sigma_{2,n}=-\sigma_{1,n},\\
	\tilde{\sigma}_{2,n}
	&=\frac{\p P}{\p w^2}\sigma_{1,n}+\frac{\p Q}{\p w^2}\sigma_{2,n}=Q(\sigma_{1,n}+\sigma_{2,n}).
	\end{align*}
Then by using \eqref{4.1}, we have 
	\begin{align}
		\frac{\p w^{1}}{\p \tilde{s}_i}&=(\Lambda-1)\left(\frac{\p f_{2,0}}{\p \tilde{s}_i}\right),\quad \frac{\p w^2}{\p \tilde{s}_i}=(1-\Lambda^{-1})\tilde{B}_+\left(\frac{\p f_{1,0}}{\p \tilde{s}_m}\right),\label{4.1.2}\\
		\frac{\p\tilde{\sigma}_{2,n}}{\p \tilde{s}_i}&=\frac{\p \psi_{2,0}^n}{\p\tilde{s}_i},\quad \frac{\p\tilde{\sigma}_{1,n}}{\p \tilde{s}_i}=\tilde{B}_+\left(\frac{\p\psi_{1,0}^n}{\p \tilde{s}_i}\right),\label{4.1.3}
	\\	\frac{\p\tilde{\sigma}_{2,n}}{\p t^{\al,p}}&=\frac{\p}{\p \tau_n}\left(\frac{\delta H_{\al,p}}{\delta w^2}\right),\quad \frac{\p\tilde{\sigma}_{1,n}}{\p t^{\al,p}}=\frac{\p}{\p \tau_n}\left(\frac{\delta H_{\al,p}}{\delta w^1}\right),\label{4.1.4}
	\end{align}
	where \[\tilde{B}_+=\frac{\tilde{\Lambda}-1}{\ve\p_{t^{1,0}}}=1+\frac{\ve}{2}\frac{\p}{\p t^{1,0}}+\dots,\quad \tilde{\Lambda}=e^{\ve\frac{\p}{\p t^{1,0}}}.\]
	In the new coordinates the Hamiltonian operators are given by \eqref{bho} and their dispersionless limits coincide with the Hamiltonian operators of the Principal Hierarchy of $M_\text{AL}$. The recursion relation \eqref{super-rr} of the odd variables is invariant under the above Miura type transformation, and Lemma \ref{5.1} still holds in the new coordinates.  Next let us  check the locality condition. 
	\begin{lem}
	 The locality condition
			\begin{equation}\label{lc}
		\left[\frac{\p}{\p \tilde{s}_{i}},\frac{\p}{\p t^{\al,p}}\right]\tilde{K}\in\hat{\mathcal{A}},\quad i=-1,\,2 ,\quad(\al,p)\in\mathcal{I}
		\end{equation}
holds true, here $\tilde{K}=\tilde{\sigma}_{1,0},\tilde{\sigma}_{2,0}, w^1, w^2$.
	\end{lem}
\begin{proof}
		By using the definition of the flows $\frac{\p}{\p\tilde{s}_i}$ it's easy to verify that
		\begin{equation}
	\left[\frac{\p}{\p \tilde{s}_{-1}},\frac{\p}{\p t^{\al,p}}\right]\tilde{K}\in\hat{\mathcal{A}},\quad	\left[\frac{\p}{\p \tilde{s}_{2}},\frac{\p}{\p t^{\al,p}}\right]X\in\hat{\mathcal{A}},\quad (\al,p)\in\mathcal{I}
	\end{equation}
	for $\tilde{K}=\tilde{\sigma}_{1,0},\tilde{\sigma}_{2,0}, w^1, w^2$ and $X=w_1,\,w_2$. To prove the lemma, we only need to check that
	\begin{equation}
	\left[\frac{\p}{\p \tilde{s}_{2}},\frac{\p}{\p t^{\al,p}}\right]
	\tilde{\sigma}_{i,0}\in\hat{\mathcal{A}},\quad i=1,2.	
	\end{equation}
	Let us focus on the non-local part of the flows. From \eqref{4.1.2}--\eqref{4.1.4} it follows that
	\begin{align}
\frac{\p}{\p t^{\al,p}}	\left(\frac{\p\tilde{\sigma}_{2,0}}{\p\tilde{s}_2}\right)&=c_0\frac{\p \tilde{\sigma}_{2,2}}{\p t^{\al,p}}+2f_{2,0}\frac{\p^2\tilde{\sigma}_{2,0}}{\p t^{\al,p}\p t^{2,0}}+\tilde{E}_2\left(\frac{\p \tilde{\sigma}_{2,0}}{\p t^{\al,p}}\right)+loc,\\
	\frac{\p}{\p \tilde{s}_2 }\left(\frac{\p \tilde{\sigma}_{2,0}}{\p t^{\al,p}}\right)&=\frac{\p^2}{\p \tau_0\p\tilde{s}_2}\left(\frac{\delta H_{\al,p}}{\delta w^2}\right)
	=c_0\frac{\p\tilde{\sigma}_{2,2}}{\p t^{\al,p}}+2f_{2,0}\frac{\p^2\tilde{\sigma}_{2,0}}{\p ^{\al,p}\p t^{2,0}}+\tilde{E}_2\left(\frac{\p \tilde{\sigma}_{2,0}}{\p t^{\al,p}}\right)+loc,
	\end{align}
	where $loc\in\hat{\mathcal{A}}$ and 
	\begin{align*}
			\tilde{E}_2=&\sum_{k\ge 1}\frac{(2+k)!}{(k-1)!}\left(t^{0,k}\frac{\partial}{\partial t^{0,k+2}}-t^{0,-k-2}\frac{\partial}{\partial t^{0,-k}}+t^{1,k}\frac{\partial}{\partial t^{1,k+2}}+t^{2,k-1}\frac{\partial}{\partial t^{2,k+1}}\right)\notag\\
		&+\sum_{k\ge 0}\left(3k^2+6k+2\right)\left(2t^{1,k}+t^{0,k}\right)\frac{\partial}{\partial t^{2,k+1}}-t^{0,-1}\frac{\partial}{\partial t^{2,0}}+\sum_{k\ge 0}(k+c_0)\tau_k\frac{\partial}{\partial \tau_{k+2}}.
		\end{align*}
The analogue of the above relation for $\tilde{\sigma}_{1,0}$ can be proved similarly, therefore we arrive at
	\begin{equation}
	\left[\frac{\p}{\p \tilde{s}_{2}},\frac{\p}{\p t^{\al,k}}\right]
	\tilde{\sigma}_{i,0}\in\hat{\mathcal{A}},\quad i=1,2.
	\end{equation}
	The Lemma is proved.
\end{proof}

Now let us proceed to prove the validity of \eqref{31-w1}.  
\begin{lem}
The flows \eqref{zh-36} yield an infinite number of Virasoro symmetries for the tau cover \eqref{zh-33} of the extended Ablowitz-Ladik hierarchy, i.e., the relations \eqref{31-w1} hold true.
\end{lem}
\begin{proof}
From the definition of the odd flows given in Proposition \ref{sp-extension}of the super tau cover of the extended Ablowitz-Ladik hierarchy we know that
	\begin{equation}
	\left[\frac{\partial }{\partial \tau_{0}},\frac{\partial}{\partial t^{\alpha,p}}\right]K=\left[\frac{\partial }{\partial \tau_{1}},\frac{\partial}{\partial t^{\alpha,p}}\right]K=0,\quad(\al,p)\in\mathcal{I}
	\end{equation}
for $K=f_{\beta,q}, P, Q,\sigma_{1,n},\sigma_{2,n}$,	and from Lemma \ref{5.1} we have
	\begin{equation}
	\left[\frac{\p}{\p \tilde{s}_{i}},\frac{\p}{\p\tau_j}\right]K=0,\quad i=-1, 2,\, j=0, 1.\label{w5}
	\end{equation}
Thus by using the Jacobi identity we can obtain
	\begin{equation}\label{zh-42}
	 \left [\frac{\partial }{\partial \tau_{j}},\left[\frac{\partial}{\partial \tilde{s}_{i}},\frac{\partial}{\partial t^{\alpha,p}}\right]\right]K=-\left(\left [\frac{\partial }{\partial \tilde{s}_i},\left[\frac{\partial}{\partial t^{\al,p}},\frac{\partial}{\partial \tau_{j}}\right]\right]+\left [\frac{\partial }{\partial t^{\al,p}},\left[\frac{\partial}{\partial \tau_j},\frac{\partial}{\partial \tilde{s}_{i}}\right]\right]\right)K=0.
	\end{equation}
By uisng the locality condition \eqref{lc} and the Virasoro symmetries of the Principal Hierarchy of $M_{\text{AL}}$ we know that $\left[\frac{\p}{\p s_{-1}},\frac{\p}{\p t^{\al,p}}\right], \left[\frac{\p}{\p s_{2}},\frac{\p}{\p t^{\al,p}}\right]$ are derivations on $\hat{\mathcal{A}}$ whose differential degrees are greater than two, thus from \eqref{zh-42} and the triviality of certain variational bihamiltonian cohomologies \cite{liu2023variational, 2109.01845} associated with the bihamiltonian structure $(\mathcal{P}_0, \mathcal{P}_1)$ we arrive at, following the argument given in the proof of existence and uniqueness of deformations of Virasoro symmetries of the Principal Hierarchy of a semisimple Frobenius manifold in \cite{2109.01845}, a proof of the lemma.
\end{proof}

So far we have proved that the extended Ablowitz-Ladik hierarchy possesses Virasoro symmetries which are represented by actions of the Virasoro operators $L_m\, (m\ge -1)$ on the tau function, as it is shown by \eqref{zh-36}. On the other hand, the topological deformation of the Principal Hierarchy of the generalized Frobenius manifold $M_\textrm{AL}$ is uniquely determined by the requirement that the Virasoro symmetries of the Principal Hierarchy act linearly on the tau function of the deformed integrable hierarchy which has the genus expansion form given by \eqref{zh-2}.
Since the dispersionless limit of the extended Ablowitz-Ladik hierarchy coincides with the Principal Hierarchy of the generalized Frobenius manifold $M_\textrm{AL}$, in order to prove the coincidence of these two integrable hierarchies, we only need to show that the tau function of the extended Ablowitz-Ladik hierarchy also has the genus expansion of the form given in \eqref{zh-2}. 

Let us denote by $\tilde{w}^1, \tilde{w}^2$ the normal coordinates of the extended Ablowitz-Ladik hierarchy \cite{normal}, they are related with the tau function of the integrable hierarchy by
\[\tilde{w}^{\al}=\ve^2\eta^{\al\beta}\frac{\p^2\log\tau}{\p x\p t^{\beta,0}},\quad \al=1,2.\]
From \cite{dubrovin2006hamiltonian} we know that any semisimple bihamiltonian integrable hierarchy is quasi-trivial, so there exists a quasi-Miura transformation 
\begin{equation}
	\tilde{w}^{\al}=v^{\al}+g^{\al}\quad \al=1,2
	\end{equation}
which reduces the extended Ablowitz-Ladik hierarchy to its dispersionless limit, here the functions $g_{\al}:=\eta_{\al\beta}g^{\beta}$ can be represented in the form
	\begin{equation*}
	g_{\al}=\sum\limits_{k\ge1}\ve^k \Phi_{\al}^k(v;v_x,\dots,v^{(n_k)})
	\end{equation*}
with $\Phi_\al^k$ depending rationally on $\p_x v^\al,\dots, \p_x^{n_k} v^\al$ and having differential degree $k$. It can be checked that in the normal coordinates the bihamiltonian structure $(\mathcal{P}_0, \mathcal{P}_1)$ does not contain linear in $\epsilon$ terms, so $\Phi_{\al}^1=0$.	 Since both $\tilde{w}^\al$ and $v^\al$ are Casimirs of the first Hamiltonian structure of the dispersionless extended Ablowitz-Ladik hierarchy, from Lemma 3.7.21 of \cite{normal} it follows the existence of functions $G_{\al}(v,v_x,\dots)$ with differential degree greater than one  such that 
	\begin{equation*}
	g_{\al}=\p_x G_{\al}, \quad\al=1,2.
	\end{equation*}
On the other hand, from the definition of the normal coordinates we know that
	\begin{equation*}
	\frac{\p g_{\al}}{\p t^{\beta,0}}=\frac{\p g_{\beta}}{\p t^{\al,0}},\quad \al,\beta=1,2.
	\end{equation*}
These relations can be rewritten as 
	\begin{equation*}
	\frac{\p}{\p x}\left(\frac{\p G_{\al}}{\p t^{\beta,0}}-\frac{\p G_{\beta}}{\p t^{\al,0}}\right)=0,\quad \al,\beta=1,2.
	\end{equation*}
	Since $\deg_x\left(\frac{\p G_{\al}}{\p t^{\beta,0}}-\frac{\p G_{\beta}}{\p t^{\al,0}}\right)\ge2$, we obtain from the above relations that 
	\begin{equation}\label{zh-43}
	\frac{\p G_{\al}}{\p t^{\beta,0}}=\frac{\p G_{\beta}}{\p t^{\al,0}},\quad \al,\beta=1,2.
	\end{equation}
Now let us perform a linear reciprocal transformation on the dispersionless limit of the Ablowitz-Ladik hierarchy by exchanging the spatial variable $x$ and the time variable $t^{1,0}$, and denote $\tilde{x}=t^{1,0}$. Then we have
	\begin{equation}\label{eq1}
	\frac{\p G_1}{\p t^{2,0}}=\frac{\p G_{\beta}}{\p \tilde{x}},
	\end{equation}
	so $G_1$ is a conserved quantity of the flows $\frac{\p}{\p t^{\beta,0}}$ with differential degree greater than one. Thus it follows from Theorem A.3 of \cite{dubrovin2018bihamiltonian} that there exists a function $G(v, v_x,\dots)$  such that 
	\begin{equation*}
	G_1=\frac{\p G}{\p \tilde{x}}, 
	\end{equation*}
	so the equation \eqref{eq1} can be rewritten as 
	\begin{equation*}
	\frac{\p}{\p \tilde{x}}\left(\frac{\p G}{\p t^{2,0}}-G_2\right)=0.
	\end{equation*}
	Since the differential degree of $\left(\frac{\p G}{\p t^{2,0}}-G_{2}\right)$ is greater than one, we arrive at
	\begin{equation*}
	G_2=\frac{\p G}{\p t^{2,0}}.
	\end{equation*}
	Therefore we have 
	\begin{equation}
	\tilde{w}^{\al}=v^{\al}+\eta^{\al\beta}\frac{\p^2 G}{\p x\p t^{\beta,0}}.
	\end{equation}
	where the function $G$ can be represented in the form 
	\[G=\sum\limits_{k\ge 2}\ve^{k}\mathcal{F}^{[k]}(v;v_x,\dots,v^{(m_k)}).\] 
	Now by using \eqref{tc} and \eqref{tauf} we obtain
	\begin{align*}
	&\eta^{\lambda\gamma}\frac{\p^2}{\p x\p t^{\gamma,0}}(	\frac{\p^2 G}{\p t^{\al,p}\p t^{\beta,q}})=\frac{\p^2}{\p t^{\al,p}\p t^{\beta,q}}(\tilde{w}^{\lambda}-v^{\lambda})\\
	=&\eta^{\lambda\gamma}\frac{\p^2}{\p x\p t^{\gamma,0}}\left(\Omega_{\al,p;\beta,q}(\tilde{w})-\Omega^{[0]}_{\al,p;\beta,q}(v)\right).
	\end{align*}
	Since $\deg_x\left(\frac{\p^2 G}{\p t^{\al,p}\p t^{\beta,q}}-\Omega_{\al,p;\beta,q}(\tilde{w})+\Omega^{[0]}_{\al,p;\beta,q}(v)\right)\ge1$, it follows from th eabove relation that 
	\[\frac{\p^2 G}{\p t^{\al,p}\p t^{\beta,q}}=\Omega_{\al,p;\beta,q}(\tilde{w})-\Omega^{[0]}_{\al,p;\beta,q}(v),\quad \forall \,(\al,p),\,(\beta,q)\in\mathcal{I},\]
	Thus we have
		\begin{equation}
	\ve^2\frac{\p^2( \ve^{-2}\log\tau^{[0]}+\ve^{-2}G)}{\p t^{\al,p}\p t^{\beta,q}}=\Omega_{\al,p;\beta,q}(\tilde{w})=\ve^2\frac{\p^2\log\tau}{\p t^{\al,p}\p t^{\beta,q}},
	\end{equation}
and the tau function of the extended Ablowitz-Ladik hierarchy has the genus expansion of the form \eqref{zh-2}. Thus we have the following theorem, and from which the Main Theorem given in Sec.\,1 also follows.
\begin{theorem}
The topological deformation of the Principal Hierarchy of the generalized Frobenius manifold $M_\textrm{AL}$ coincides with the extended Ablowitz-Ladik hierarchy \eqref{w-1}.
\end{theorem}

%
%

\section{Conclusion}
In this paper, we present an explicit construction of logarithmic flows of the Principal Hierarchy of the generalized Frobenius manifold  $M_\textrm{AL}$. We then define a certain extension of the Ablowitz-Ladik hierarchy, and show that it is a tau-symmetric deformation of this Principal Hierarchy. We prove that this extended Ablowitz-Ladik hierarchy possesses an infinite number of Virasoro symmetries which act linearly on its tau function, and that it coincides with the topological deformation of the Principal Hierarchy of $M_\textrm{AL}$.  

Unlike the positive and negative flows of the extended Ablowitz-Ladik hierarchy, for the flows that correspond to the logarithmic flows of the Principal Hierarchy of $M_\textrm{AL}$, which we call the logarithmic flows of the extended Ablowitz-Ladik hierarchy, we do not give their construction explicitly in terms of the Lax operator $L$ in the present paper. It would be interesting to give a definition of the logarithm of the Lax operator $L$ and to represent these flows in terms of $\log L$. On the other hand, for the logarithmic flow $\frac{\partial}{\partial t^{1,0}}$ we obtain in Theorem \ref{zh-45} an explicit expression of the exponential of the flow $\frac{\p}{\p t^{1,0}}-\frac{\p}{\p x}$, which yields an auto-B\"acklund transformation of the extended Ablowitz-Ladik hierarchy. We expect that exponentials of other logarithmic flows also have
simple and explicit expressions.

\vskip 0.3truecm
\noindent\textbf{Acknowledgement.}\, This work is supported by NSFC No.\,12171268. The authors would like to thank Di Yang and Zhe Wang for very helpful discussions on this work.


\vskip 0.3truecm 

\begin{thebibliography}{99}
	
\bibitem{1973Nonlinear}
M.J. Ablowitz, D.J. Kaup, A.C. Newell, H. Segur, 
Nonlinear-Evolution Equations of Physical Significance,
Phys. Rev. Lett. 31 (1973) 125--127.

\bibitem{ablowitz1975nonlinear}
M.J. Ablowitz, J.F. Ladik, Nonlinear differential-difference equations,
J. Math. Phys. 16 (1975) 598--603.
	
\bibitem{M1976Nonlinear}
M.J. Ablowitz, , J.F. Ladik,
Nonlinear differential–difference equations and {F}ourier analysis,
J. Math. Phys. 17 (6) (1976) 1011--1018.

\bibitem{suris2012problem}
Y.B. Suris, 
The problem of integrable discretization: Hamiltonian approach, in: Progress in Mathematics vol.219, Birkh\"{a}user Verlag, Basel, 2003.
	
\bibitem{brini2012local}
A.Brini, 
The local {G}romov-{W}itten theory of {$\Bbb{C}\Bbb{P}^1$} and integrable hierarchies,
Comm. Math. Phys. 313 (3) (2012) 571--605.
	
\bibitem{Brini2011Integrable}
 A.Brini, G. Carlet, P.Rossi, 
Integrable hierarchies and the mirror model of local {$\Bbb{C}\Bbb{P}^1$},
 Phys. D 241 (23-24) (2012) 2156–2167.
	
\bibitem{carlet2003extended}
G. Carlet, B. Dubrovin , Y Zhang.
The extended {T}oda hierarchy, 
Mosc.Math.J. 4 (2) (2004) 313--332.
	
\bibitem{magri} 
L.Degiovanni, F.Magri, V.Sciacca, 
On deformation of {P}oisson manifolds of hydrodynamic type, 
Comm. Math. Phys.253 (1)(2005) 1--24.
	
\bibitem{Du-94}
B. Dubrovin, 
 Geometry of {$2$}D topological field theories,in: M. Francaviglia, S. Greco (Eds.), Integrable Systems and Quantum Groups, in: Lecture Notes in Mathematics, vol.1620, Springer, Berlin, Heidelberg, 1996.
 
\bibitem{dubrovin1999painleve}
B.Dubrovin, 
Painlev{\'e} transcendents in two-dimensional topological field theory,
in: R. Conte (Ed.), The Painlev{\'e} Property, in: CRM Series in Mathematical Physics, Springer, New York, NY, 1999.
	
\bibitem{2008On}
B.Dubrovin, 
On universality of critical behaviour in Hamiltonian PDEs, in:V. M. Buchstaber, I. M. Krichever(Eds.),Geometry, topology, and mathematical physics, in: American Mathematical Society Translations, Series 2,vol.224, American Mathematical Society, Providence, RI, 2008.

	
\bibitem{dubrovin2006hamiltonian}
B. Dubrovin,  S.-Q. Liu,  Y. Zhang,
On Hamiltonian perturbations of hyperbolic systems of conservation laws {I}: Quasi-triviality of bi-Hamiltonian perturbations,
 Comm. Pure Appl. Math. 59  (4) (2006) 559--615.
	
\bibitem{dubrovin2018bihamiltonian}
B. Dubrovin,  S.-Q. Liu,  Y. Zhang,
Bihamiltonian cohomologies and integrable hierarchies {II}: {T}he tau structures,
Comm. Math. Phys. 361 (2) (2018) 467--524.
	
\bibitem{1999Frobenius}
B. Dubrovin, Y. Zhang,
Frobenius manifolds and Virasoro constraints,
 Selecta Math. (N.S.) 5 (4) (1999) 423--466.
	
\bibitem{normal}
B. Dubrovin, Y. Zhang,
Normal forms of hierarchies of integrable PDEs, Frobenius manifolds and Gromov-Witten invariants,
	arXiv preprint math/0108160 (2001).
	
\bibitem{eguchi1994topological}
T. Eguchi,  S.K. Yang,
The topological {${\bf C}{\rm P}^1$} model and the large-{$N$} matrix integral,
Modern Phys. Lett. A 9 (31) (1994) 2893--2902.
	
\bibitem{ercolani2006bi}
N.M.Ercolani,  G.I.Lozano, 
A bi-{H}amiltonian structure for the integrable, discrete non-linear {S}chr{\"o}dinger system,
Phys. D 218 (2) (2006) 105--121.
	
\bibitem{gekhtman2009multi}
M.Gekhtman, I.Nenciu,
Multi-{H}amiltonian structure for the finite defocusing {A}blowitz-{L}adik equation,
Comm. Pure Appl. Math. 62 (2) (2009) 147--182.
	
\bibitem{getzler}
E. Getzler, 
A Darboux theorem for Hamiltonian operators in the formal calculus of variations, 
Duke Math. J. 111 (3) (2002) 535--560.
	
	
\bibitem{li2022tri}
 S. Li, S.-Q. Liu, H. Qu, Y. Zhang, 
 Tri-Hamiltonian Structure of the Ablowitz-Ladik Hierarchy,
Phys. D 433 (2022) 133180, 17pp .
	
\bibitem{liu2018lecture}
 S.-Q. Liu,
 Lecture Notes on Bihamiltonian Structures and Their Central Invariants, in: E. Clader, Y. Ruan (Eds.), B-model {G}romov-{W}itten theory, in: Trends in Mathematics, Birkh\"{a}user/Springer, Cham, 2018.

	
\bibitem{liu2022generalized}
S.-Q. Liu, H. Qu, Y. Zhang, 
Generalized {F}robenius manifolds with non-flat unity and integrable hierarchies,
arXiv preprint arXiv:2209.00483 (2022).
	
\bibitem{liu2024}
 S.-Q. Liu, H. Qu, Y. Wang, Y. Zhang, 
Solutions of the loop equations of a class of generalized {F}robenius manifolds,
 arXiv preprint arXiv: 2402.00373 (2024).
	
\bibitem{liu2021super}
S.-Q. Liu, Z. Wang, Y. Zhang, 
Super tau-covers of bihamiltonian integrable hierarchies.
J. Geom. Phys. 170 (2021), 104351, 25pp.
	
\bibitem{2109.01845}
S.-Q. Liu, Z. Wang, Y. Zhang, 
Variational bihamiltonian cohomologies and integrable hierarchies {II}: {V}irasoro symmetries,
Comm. Math. Phys. 395 (1) (2022) 459--519.
	
\bibitem{liu2023variational}
S.-Q. Liu, Z. Wang, Y. Zhang, 
Variational bihamiltonian cohomologies and integrable hierarchies {I}: {F}oundations,
Comm. Math. Phys. 401 (1) (2023) 985--1031.

	
\bibitem{reci}
S.-Q. Liu, Z. Wang, Y. Zhang, 
{V}ariational bihamiltonian cohomologies and integrable hierarchies {III}: {L}inear reciprocal transformations,
Comm. Math. Phys. 403 (2) (2023) 1109--1152.
	
\bibitem{liu2011jacobi}
S.-Q. Liu, Y. Zhang, 
Jacobi structures of evolutionary partial differential equations,
Adv. Math. 227  (1) (2011) 73--130.
	
\bibitem{liu2013bihamiltonian}
S.-Q. Liu, Y. Zhang, 
Bihamiltonian cohomologies and integrable hierarchies {I}: {A} special case,
Comm. Math. Phys. 324 (3) (2013) 897--935.
	
\bibitem{oevel1989mastersymmetries}
W. Oevel, B. Fuchssteiner, H. Zhang,  O.Ragnisco, 
Mastersymmetries, angle variables, and recursion operator of the relativistic {T}oda lattice,
J. Math. Phys. 30  (11) (1989) 2664--2670.
	


\end{thebibliography}
\end{document}